\documentclass[twocolumn,pra,notitlepage,superscriptaddress,tightenlines,preprintnumbers]{revtex4}
\usepackage{url}
\usepackage{braket}
\usepackage{graphicx}
\usepackage{multirow}
\newcommand{\be}{\begin{equation}}
\newcommand{\ee}{\end{equation}}
\newcommand{\ba}{\begin{array}}
	\newcommand{\ea}{\end{array}}
\newcommand{\bea}{\begin{eqnarray}}
\newcommand{\eea}{\end{eqnarray}}

\usepackage{amsthm}
\usepackage{amsmath}
\usepackage{amssymb}
\usepackage{soul}
\usepackage{verbatim}
\usepackage{bm}
\usepackage{outlines}
\usepackage{siunitx}
\usepackage[space]{grffile}
\usepackage[colorlinks=true, linkcolor=blue, citecolor=gray]{hyperref}
\usepackage[capitalize]{cleveref}
\usepackage{xr}
\usepackage{color}
\usepackage{tikz}
\usetikzlibrary{calc, positioning}
\usepackage{mathtools}
\usepackage{algorithm}
\usepackage{algorithmic}
\newtheorem{theorem}{Theorem}

\newtheorem{lemma}[theorem]{Lemma}

\usepackage{url}
\usepackage{braket}
\usepackage{graphicx}
\usepackage{multirow}
\usepackage{amsmath}
\DeclareMathOperator{\tr}{Tr}
\usepackage{amsthm}
\usepackage{soul}
\usepackage{verbatim}
\usepackage{bm}
\usepackage{outlines}
\usepackage{siunitx}
\usepackage[space]{grffile}
\usepackage{color}
\usepackage{tikz}
\usetikzlibrary{calc, positioning}
\usepackage{mathtools}

\newcommand{\abs}[1]{\left|{#1}\right|}
\newcommand{\norm}[1]{\left|\left|{#1}\right|\right|}

\newcommand{\PTMLindblad}{L^\mathcal{L}}
\newcommand{\hcritical}{h_x^\mathrm{cr}}

\begin{document}
	
	\title{Spectral Quantum Tomography}
	\author{Jonas Helsen}
	\email{j.helsen@tudelft.nl}
	\affiliation{QuTech, Delft University of Technology, P.O. Box 5046, 2600 GA Delft, The Netherlands}
	
	\author{Francesco Battistel}
	\affiliation{QuTech, Delft University of Technology, P.O. Box 5046, 2600 GA Delft, The Netherlands}
	
	\author{Barbara M. Terhal}
	\affiliation{QuTech, Delft University of Technology, P.O. Box 5046, 2600 GA Delft, The Netherlands}
	\affiliation{JARA Institute for Quantum Information, Forschungszentrum Juelich, D-52425 Juelich, Germany}
	\date{\today}
	\begin{abstract}
		We introduce spectral quantum tomography, a simple method to extract the eigenvalues of a noisy few-qubit gate, represented by a trace-preserving superoperator, in a SPAM-resistant fashion, using low resources in terms of gate sequence length. The eigenvalues provide detailed gate information, supplementary to known gate-quality measures such as the gate fidelity, and can be used as a gate diagnostic tool. We apply our method to one- and two-qubit gates on two different superconducting systems available in the cloud, namely the QuTech Quantum Infinity and the IBM Quantum Experience. We discuss how cross-talk, leakage and non-Markovian errors affect the eigenvalue data.
	\end{abstract}
	\maketitle

	\section{Introduction}
	A central challenge on the path towards large-scale quantum computing is the engineering of high-quality quantum gates. To achieve this goal, many methods which accurately and reliably characterize quantum gates have been developed. Some of these methods are scalable, meaning that they require an effort which scales polynomially in the number of qubits on which the gates act. Scalable protocols, such as randomized benchmarking~\cite{knill+:RB, MGE:RB,magesan+:IRB,WG:leak,wallman+:unitarity, DHW:URB, erhard+:cycle, OWE:RB} necessarily give a partial characterization of the gate quality, for example an average gate fidelity. Other protocols such as robust tomography~\cite{kimmel+:tomo} or gate-set tomography~\cite{BK:GST,greenbaum:GST} trade scalability for a more detailed characterization of the gate. A desirable feature of all the above protocols is that they are resistant to state-preparation and measurement (SPAM) errors. The price of using SPAM-resistant (scalable) methods is that these protocols have significant experimental complexity and/or require assumptions on the underlying hardware to properly interpret their results.
	
	In this work we present spectral quantum tomography, a simple non-scalable method that extracts spectral information from noisy gates in a SPAM-resistant manner. 
	To process the tomographic data and obtain the spectrum of the noisy gate, we rely on the matrix-pencil technique, a well-known classical signal processing method. This technique has been advocated in~\cite{OWE:RB} in the context of randomized benchmarking and has also been used in~\cite{OTT:QPE} for processing data in the algorithm of quantum phase estimation. It has also been used, under the phrase `linear systems identification', in~\cite{BL:system} to predict the time evolution of quantum systems. While the matrix pencil technique leads to explicitly useful estimates of eigenvalues and their amplitudes, we note that the same underlying idea is used in the method of ``delayed vectors" which has been proposed in \cite{WP:dim} to assess the dimensionality of a quantum system from its dynamics. This ``delayed vectors" approach has been applied to assess leakage in superconducting devices in \cite{SDK:leak-witness}. 

	The spectral information of a noisy gate $\mathcal{S}$, which approximates some target unitary $U$, is given by the eigenvalues of the so-called Pauli transfer matrix representing ${\cal S}$. These eigenvalues, which are of the form $\lambda=\exp(-\gamma)\exp(i \phi)$, contain information about the quality of the implemented gate. Intuitively, the parameter $\gamma$ captures how much the noisy gate deviates from unitarity due to entanglement with an environment, while the angle $\phi$ can be compared to the rotation angles of the targeted gate $U$. Hence $\phi$ gives information about how much one over- or under-rotates.  The spectrum of~$\mathcal{S}$ can also be related to familiar gate-quality measures such as the average gate fidelity and the unitarity. Moreover, in the case of a noisy process modeled by a Lindblad equation, the spectrum can be easily related to the more familiar notions of relaxation and dephasing times. 

	The main advantage of spectral quantum tomography is its simplicity, requiring only the (repeated) application of a single noisy gate $\mathcal{S}$, as opposed to the application of a large set of gates as in randomized benchmarking, gate-set tomography and robust tomography. Naturally, simplicity and low-cost come with some drawback, namely the method does not give information about the eigenvectors of the noisy gate, such as the axis around which one is rotating. However, information about the eigenvectors is intrinsically hard to extract in a SPAM-resistant fashion since SPAM errors can lead to additional rotations \cite{RPZ:gauge}.
	Another feature of spectral quantum tomography is that it can be used to extract signatures of non-Markovianity, namely the phenomenon where the noisy gate $\mathcal{S}$ depends on the context in which it is applied (e.g. time of application, whether any gates have been applied before it). As we show in this paper, our method can be used to detect various types of non-Markovian effects such as coherent revivals, parameter drifts, and Gaussian-distributed time-correlated noise. It is also possible to distinguish non-Markovian effects from qubit leakage.
	For these reasons we believe that spectral quantum tomography adds a useful new tool to the gate-characterization toolkit. The method could also have future applications in assessing the performance of logical gates in a manner which is free of logical state preparation and measurement errors, see the Discussion Section.

	\section{Results}

	\subsection{Eigenvalues of Trace-Preserving Completely-Positive (TPCP) maps}
	\label{sec:eig}
	Take a unitary gate $U$ on a $d$-dimensional space with $U \ket{\psi_j}=e^{i \phi_j}\ket{\psi_j}$. The corresponding TPCP map ${\cal S}_U(\rho)=U \rho U^{\dagger}$ has one trace-full eigenvector, namely~$I$ with eigenvalue~1, as well as $d^2-1$ traceless eigenvectors. In particular, there are $d^2-d$ traceless eigenvectors of the form $\ket{\psi_j}\bra{\psi_l}$ for $j\neq l$ with eigenvalues $\exp(i (\phi_j-\phi_l))$, and $d-1$ traceless eigenvectors of the form $\ket{\psi_1}\bra{\psi_1}-\ket{\psi_j}\bra{\psi_j}$ for $j=2,\ldots, d$ with eigenvalue~$1$.

	For general TPCP maps it is convenient to use the Pauli transfer matrix formalism.
	For an $n$-qubit system ($d=2^n$) consider the normalized set of Pauli matrices $P_{\mu}$ for $\mu=0,\ldots, N$ with $N+1=4^n=d^2$, where $P_0=I/\sqrt{2^n}$ and the normalization is chosen such that ${\rm Tr}\big[ P_{\mu} P_{\nu}\big]=\delta_{\mu \nu}$. For a TPCP map ${\cal S}$ acting on $n$ qubits, the Pauli transfer matrix is then defined as
	\begin{equation}
	S_{\mu \nu}={\rm Tr}\big[P_\mu {\cal S}(P_{\nu})\big],\;\; \mu,\nu=0,\ldots, N.
	\label{def:Pauli}
	\end{equation}
	The form of the Pauli transfer matrix $S$ is~\cite{RSW:TCP}
	\begin{equation}
	{\cal S} \leftrightarrow S=\left(\begin{array}{cc} 1 & 0 \\
	{\bf s} & T^{\cal S} \\ 
	\end{array}\right),
	\label{eq:superop}
	\end{equation}
	where $T^{\cal S}$ is a real $N\times N$ matrix and ${\bf s}$ is a $N$-dimensional column vector. The 1 and 0's in the top row of the Pauli transfer matrix are due to the fact that ${\cal S}$ is trace-preserving. For a unital ${\cal S}$ which obeys ${\cal S}(I)=I$, the vector ${\bf s}=0$. 
	
	A few properties are known of the eigenvalue-eigenvector pairs of $S$, i.e.~the pairs $(\lambda, \vec{\bf v})$ with $S{\bf v}=\lambda {\bf v}$: 
	\begin{itemize}
		\item The eigenvalues of $S$ are $1 $ and the eigenvalues of $T^{\cal S}$ since the solutions of the equation ${\rm det}(S-\lambda I)=0$ are the solutions of the equation $(1-\lambda){\rm det}(T^{\cal S}-\lambda I)=0$.
		\item The eigenvalues of $S$, and thus the eigenvalues of $T^{\cal S}$, come in complex-conjugate pairs. This is true because $T^{\cal S}$ is a real matrix.
		\item The eigenvalues of $T^{\cal S}$ (or $S$ for that matter) have modulus less than 1, i.e.~$|\lambda| \leq 1$ (see e.g.~Proposition 6.1 in~\cite{wolf}). 
	\end{itemize}
	
	If $T^{\cal S}$ is {\em diagonalizable} as a matrix, it holds that $T^{\cal S}=V D V^{-1}$ where $D$ is a diagonal matrix and $V$ a similarity transformation. Generically, $T^{\cal S}$ will be diagonalizable, in which case there are $N$ eigenvalue-eigenvector pairs for $T$. A sufficient condition for diagonizability is, for example, that all the eigenvalues of $T^{\cal S}$ are distinct.
	In Appendix \ref{sec:nondiag} we give examples and discuss what it means if $T^{\cal S}$ is not diagonalizable.
	
	For some simple single-qubit channels we can explicitly compute the spectrum. For instance, for a single-qubit depolarizing channel with depolarizing probability $p$, the eigenvalues of the sub-matrix $T^{\cal S}$ of the Pauli transfer matrix are $\{1-p, 1-p, 1-p\}$. For a single qubit amplitude-damping channel with damping rate $p$ they are $\{\sqrt{1-p},\sqrt{1-p},1-p\}$~\cite{greenbaum:GST}.
	
	\subsubsection{Relation to gate-quality measures}
	\label{sec:gate_quality_bounds}
	The eigenvalues of the Pauli transfer matrix of a noisy gate ${\cal S}$ can be related to several other known measures of gate quality such as the average gate fidelity ${\cal F}(\mathcal{S},U)$, the gate unitarity $u({\cal S})$ and, for a single qubit ($n=1$), the gate unitality.
	
	The average gate fidelity is defined as ${\cal F}(\mathcal{S},U)=\int d\phi \bra{\phi} U^{\dagger} {\cal S}(\ket{\phi}\!\bra{\phi}) U \ket{\phi}$. This fidelity relates directly to the entanglement fidelity ${\cal F}_{\rm ent}(\mathcal{S},U)$ via ${\cal F}=\frac{{\cal F}_{\rm ent}d+1}{d+1}$~\cite{HHH:singlet}, where the entanglement fidelity is defined as
	\begin{equation*}
	{\cal F}_{\rm ent}(\mathcal{S},U)={\rm Tr} \big[I \otimes U \ket{\Psi}\!\bra{\Psi} I \otimes U^{\dagger}(I\otimes {\cal S})(\ket{\Psi}\!\bra{\Psi})\big],
	\end{equation*}
	where $\ket{\Psi} = \frac{1}{\sqrt{d}}\sum_{i=1}^d\ket{i,i}$ is a maximally entangled state.
	Using that $\ket{\Psi}\bra{\Psi}=\frac{1}{d} \sum_{\mu=0}^{N} P_{\mu} \!\otimes\! P_{\mu}$  and $U P_{\mu} U^\dagger=\sum_{\kappa} T_{\mu \kappa}^{U^\dagger} P_{\kappa}$ we can write
	\begin{equation*}
	{\cal F}_{\rm ent}(\mathcal{S},U)=\frac{1}{d^2} \sum_{\mu} {\rm Tr}\big[U P_{\mu} U^{\dagger} {\cal S}(P_{\mu})\big]=\frac{1}{d^2}\big(1+\tr \big[T^{U^\dagger}T^{\cal S} \big]\big).
	\end{equation*}
	Thus for the (entanglement) fidelity of a noisy gate ${\cal S}$ with respect to the identity channel $U=I$, one has ${\cal F}_{\rm ent}(\mathcal{S},I)=\frac{1}{d^2}(1+\sum_i \lambda_i)$, implying a direct relation to the spectrum $\{\lambda_i\}$ of $T^{\cal S}$. A more interesting relation is how the eigenvalues of $T^{\cal S}$ bound the fidelity with respect to a targeted gate $U$. In Appendix \ref{sec:fidel} we prove that the entanglement fidelity
	can be upper bounded as
	\begin{equation}
	{\cal F}_{\rm ent}({\cal S}, U)\!\leq\! \frac{1}{d^2}\!\!\left[\!1\!+\!(d^2\!-\!1)\, \!
	\!\!\left(\sqrt{1-\frac{\sum_j |\lambda_j|^2}{d^2-1}}\!+\!\xi_{\rm max} \!\right)\!\right],
	\end{equation}
	where $\xi_{\rm max} = \frac{1}{d^2-1}|\sum_j\lambda^{\rm ideal}_j\lambda_j^*|$ with $\lambda^{\rm ideal}_j$ the eigenvalues of $T^{U}$ with $U$ the targeted unitary, ordered such that the sum $|\sum_j\lambda^{\rm ideal}_j\lambda_j^*|$ is maximal.
	
	This upper bound is not particularly tight, but for the case of a single qubit we can make a much stronger numerical statement, see~\cref{sec:fidel}. 

Another measure of gate quality, namely the unitarity or the coherence of a channel~\cite{wallman+:unitarity} on a $d$-dimensional system, is defined as 
	\begin{equation}
	u({\cal S})=\frac{d}{d-1}\int d\phi \;{\rm Tr}\big[[{\cal S}'(\ket{\phi}\!\bra{\phi})]^{\dagger} {\cal S}'(\ket{\phi}\!\bra{\phi})\big],
	\end{equation}
	where $\mathcal{S}'(\rho)\coloneqq \mathcal{S}(\rho)-\tr [\mathcal{S}(\rho)] I/\sqrt{d}$. 
	A more convenient but equivalent definition is
	\begin{equation}
	u({\cal S})=\frac{1}{d^2-1} {\rm Tr}\big[{T^{\mathcal{S}}}^\dagger{T^{\mathcal{S}}} \big] =\frac{1}{d^2-1} \sum_{i}\sigma_i(T^{\mathcal{S}})^2,
	\end{equation}
	where $\{\sigma_i\}$ are the singular values of the matrix $T^{\mathcal{S}}$.

	The unitarity captures how close the channel is to a unitary gate. A lower bound on the unitarity is given by Proposition 2 in~\cite{RPZ:gauge}:
	\begin{equation}
	u({\cal S}) \ge \frac{1+\sum_{i=1}^{d^2-1} \abs{\lambda_i}^2-d}{d(d-1)},
	\end{equation}
	where $\{\lambda_i\}$ are the eigenvalues of $T^{\mathcal{S}}$. For a single qubit, an upper bound on the unitarity can also be given in terms of a non-convex optimization problem, see~\cref{sec:fidel}.
	
	The unitality of a TPCP map is defined as $1-\norm{\mathbf{s}}^2$ with ${\bf s}$ in Eq.~(\ref{eq:superop}). Specifically, for single-qubit channels one can derive the bound~\cite{RPZ:gauge}
	\begin{equation}
	\norm{\mathbf{s}}^2 \le 1-\abs{\lambda_1}^2 -\abs{\lambda_2}^2 - \abs{\lambda_3}^2 + 2\lambda_1\lambda_2\lambda_3.
	\end{equation}
	
	\subsubsection{Relation to relaxation and dephasing times}
	\label{sec:lind}
	
	We consider the eigenvalues of a superoperator induced by a simple Lindblad equation modeling relaxation and decoherence of a driven qubit, as an example.
	We have a Lindblad equation with time-independent Lindbladian $\mathcal{L}$:
	\begin{equation}\label{eq:lindblad}
	\dot{\rho}=\mathcal{L}(\rho).
	\end{equation}
	The formal solution of~\cref{eq:lindblad} is given by $\rho(t)=e^{t\mathcal{L}}(\rho(t=0))$, where $e^{t\mathcal{L}}$ is a TPCP map for every $t$. We are interested in the total evolution after a certain gate time $\tau$ and set $\mathcal{S}_{\tau}=e^{\tau\mathcal{L}}$. We assume a simple model in which a qubit evolves according to a Hamiltonian $H=(h_x X + h_y Y + h_z Z)/2$ and is subject to relaxation and pure dephasing processes, according to the Lindbladian:
	\begin{equation*}
	\mathcal{L}(\rho)= -i[H,\rho] + \frac{1}{T_1} \Bigl(\sigma_-\rho\sigma_+-\frac{1}{2}\{\sigma_+\sigma_-,\rho\}\Bigr) + \frac{1}{2T_\phi}(Z\rho Z-\rho).
	\end{equation*}
	We define the relaxation respectively dephasing rates $\Gamma_1=1/T_1$ and $\Gamma_2=1/T_2=1/(2T_1)+1/T_\phi$. The Pauli transfer matrix $\PTMLindblad$ of $\mathcal{L}$ then takes the form
	\begin{equation}
	\PTMLindblad=\begin{pmatrix}
	0   & 0   & 0   & 0   \\
	0   & -\Gamma_2   & h_z   & h_y   \\
	0   & -h_z   & -\Gamma_2   & h_x   \\
	\Gamma_1   & -h_y   & -h_x   & -\Gamma_1   
	\end{pmatrix}.
	\label{eq:L}
	\end{equation}
	We will denote the eigenvalues of $\PTMLindblad$ by $\Omega_j$ for $j \in \{0, \ldots,3\}$ and the eigenvalues of $\mathcal{S}_\tau$ by $\lambda_j$ for $j \in \{0, \ldots,3\}$.
	As expected, $\Omega_0=0$ implying that $\lambda_0=e^0=1$ is an eigenvalue of $\mathcal{S}_\tau$. The other three eigenvalues of $\PTMLindblad$ can be found from the $3\times 3$ sub-matrix in the lower-right corner. Here we consider some simple cases.
	
	\emph{Case 1: $h_x=h_y=h_z=0$}. In this case, for $j=1,2,3$ the three eigenvalues of ${\mathcal L}$ and ${\cal S}_{\tau}$ are clearly
	\begin{gather*}
	\Omega_j \in \{-\Gamma_2, -\Gamma_2, -\Gamma_1\},\\
	\lambda_j \in \{e^{-\Gamma_2 \tau}, e^{-\Gamma_2 \tau}, e^{-\Gamma_1 \tau}\},
	\end{gather*}
	thus relating directly to the relaxation and dephasing rates.
	
	\emph{Case 2: $h_x=h_y=0$}. In this case we have
	\begin{gather*}
	\Omega_j \in \{-\Gamma_2+ih_z, -\Gamma_2-ih_z, -\Gamma_1\},\\
	\lambda_j \in \{e^{-\Gamma_2 \tau}e^{ih_z \tau}, e^{-\Gamma_2 \tau} e^{-ih_z \tau}, e^{-\Gamma_1 \tau}\},
	\end{gather*}
	where we have separated the decaying part of the $\lambda_j$ (corresponding to the real part of the $\Omega_j$) and their phases (corresponding to the imaginary part). If we work in the rotating frame of the qubit, $h_z$ can be understood as an over-rotation along the $Z$-axis, which would appear in the spectrum as an extra phase imparted to two of the eigenvalues. Again we see that the decaying part of the eigenvalues directly relates to the relaxation and dephasing rates.
	
	\emph{Case 3: $h_y=h_z=0$}. This case shows that over-rotations can also modify the decay strength of the eigenvalues. We analyze the eigenvalues as a function of $h_x$. From $\PTMLindblad$ in Eq.~(\ref{eq:L}) we see that $\Omega_1(h_x)=-\Gamma_2$ for all $h_x$. For the other eigenvalues we have
	\begin{equation}
	\Omega_{2,3}(h_x) = -\frac{1}{2} \Bigl( \Gamma_1+\Gamma_2 \pm \sqrt{(\Gamma_1-\Gamma_2)^2-4h_x^2} \Bigr).
	\end{equation}
	We see that if $|h_x|<|\Gamma_1-\Gamma_2|/2\equiv \hcritical$, only the moduli of $\lambda_2$ and $\lambda_3$ are affected as compared to \emph{Case 1}, in other words, $\lambda_2$ and $\lambda_3$ only decay with no extra phases. On the contrary, the phases of these eigenvalues becomes non-zero when the driving is sufficiently strong: $|h_x|>\hcritical$. It implies that if we look at the dynamics induced by the Lindblad equation, real oscillations, not only decay, will be present as a function of $\tau$. Hence these two scenarios represent respectively the overdamped ($|h_x| < \hcritical$) and underdamped regime ($|h_x| > \hcritical$), similar to the dynamics of a vacuum-damped qubit-oscillator system, see e.g.~\cite{book:HR}. At $|h_x|=\hcritical$, the system is critically damped and $\PTMLindblad$ does not have $4$ linearly-independent eigenvectors, meaning that the Pauli transfer matrix of ${\cal S}_{\tau}$ is not diagonalizable. In this case the dynamics also has a linear dependence on $t$ besides the exponential decay with $t$, see the discussion in Appendix \ref{sec:nondiag}.

	\subsection{Spectral tomography}
	\label{sec:stm}
	
	In this section we describe the spectral tomography method, which estimates the eigenvalues of ${\cal S}$, where ${\cal S}$ is a TPCP implementation of a targeted unitary gate. 
	
	We model state-preparation errors as a perfect preparation step followed by an unknown TPCP map ${\cal N}_{\rm prep}$. Similarly, measurement errors are modeled by a perfect measurement preceded by an unknown TPCP map ${\cal N}_{\rm meas}$. We assume that when we apply the targeted gate $k$ times, an accurate model of the resulting noisy dynamics is ${\cal S}^k$. 

The spectral tomography method can be applied without this assumption but the interpretation of its results is more difficult, see Section \ref{sec:nm} for a discussion.

	The method works by constructing the following {\em signal} function, for $k=0,1,\ldots, K$ for some fixed $K$:
	\begin{equation}
	g(k)=\sum_{\mu=1}^{N}{\rm Tr}\big[ P_{\mu}{\cal N}_{\rm meas} \circ {\cal S}^k \circ {\cal N}_{\rm prep}(P_{\mu})\big].
	\label{eq:coeff}
	\end{equation}
	Gathering the data to estimate $g(k)$ requires (1) picking a traceless $n$-qubit Pauli $P_{\mu}$, (2) preparing an $n$-qubit input state in one of the $2^n$ basis states corresponding to this chosen Pauli, (3) applying the gate $k$ times and measuring in the same chosen Pauli basis, and (4) repeating (1-3) over different Pauli's, basis states and experiments to get good statistics. As in standard process tomography~\cite{Nie00}, one takes linear combinations of the estimated probabilities for the outcomes to construct an estimator of a Pauli operator on a Pauli input. This gives an estimate of $g(k)$ for a fixed $k$.
	Repeating this process for $k\in \{0,\ldots, K\}$ we reconstruct the entire signal function. In Section \ref{sec:resources} we discuss the cost of doing these experiments as compared to randomized benchmarking.
	
	Let us now examine how $g(k)$ depends on the eigenvalues of the matrix $T$. When there are no SPAM errors, that is, ${\cal N}_{\rm meas}$ and ${\cal N}_{\rm prep}$ are identity channels, we have 
	\begin{equation}
	g^{\mbox{\tiny NO SPAM}}(k)=\sum_{\mu=1}^{N} (T^k)_{\mu \mu}={\rm Tr}\big[T^k\big]=\sum_{j=1}^N \lambda_j^k,
	\label{eq:nospam}
	\end{equation}
	where $\{\lambda_j\}$ are the eigenvalues of $T$. The last step in this equality follows directly when $T$ is diagonalizable, but it can alternatively be proved using the so-called Schur triangular form of $T$ (we give this proof in Appendix  \ref{sec:nondiag}).
	
	When ${\cal N}_{\rm meas}$ and ${\cal N}_{\rm prep}$ are not identity channels, we have
	\begin{equation}
	g(k)={\rm Tr}\big[T_{\rm meas} T^k T_{\rm prep}\big]={\rm Tr} \big[A_{\rm SPAM} D^k\big]=\sum_{j=1}^N A_j \lambda_j^k,
	\label{eq:defgk}
	\end{equation}
	where $T_{\rm meas}$ and $T_{\rm prep}$ are respectively the $T$-submatrices of the Pauli transfer matrix of ${\cal N}_{\rm meas}$ and ${\cal N}_{\rm prep}$. Here we assume that $T = VDV^{-1}$ is diagonalizable and the matrix $A_{\rm SPAM}=V^{-1} T_{\rm prep} T_{\rm meas} V$ captures the SPAM errors. One may expect that $A_{\rm SPAM}$ is close to the identity matrix in the typical case of low SPAM errors, in particular one may expect that $A_j \neq 0$ for all $j$ so that all eigenvalues of $T$ are present in the signal $g(k)$.
	
	In principle, one could take more tomographic data and consider a full matrix-valued signal $c_{\mu \nu}(k)={\rm Tr}\big[P_{\mu}{\cal N}_{\rm meas} \circ {\cal S}^k \circ {\cal N}_{\rm prep}(P_{\nu})\big]$ instead of only Eq.~(\ref{eq:coeff}). This requires doing many more experiments and there is no clear advantage in terms of the ability to determine the spectrum.

	\subsubsection{Signal analysis or matrix-pencil method for extracting eigenvalues}
	\label{sec:prony}
	
	In this section we review the classical signal-processing method which reconstructs, from the (noisy) signal 
	$g(k)=\sum_{j=1}^N A_j \lambda_j^k $ for $k=0, \ldots ,K$, an estimate for the eigenvalues $\lambda_j$ and the amplitudes $A_j$.  Note that we have $g(k) \in \mathbb{R}$ due to Eq.~(\ref{eq:coeff}).
	Not surprisingly, this signal-processing method has been employed and reinvented in a variety of scientific fields. We implement the so-called ESPRIT analysis described in~\cite{SP:pencil}, but see also~\cite{pt:prony}. 
	In the context of spectral tomography we know that the signal $g(k)$ will in principle contain $N$ eigenvalues (which are possibly degenerate). However, we can vary the number of eigenvalues we use to fit the signal to see whether a different choice than $N$ gives a significantly better fit. This is relevant in particular when the implemented gate contains leakage or non-Markovian dynamics, see Section \ref{sec:nm}. 
	\begin{figure*}[ht]
		\includegraphics[width=1.1\columnwidth]{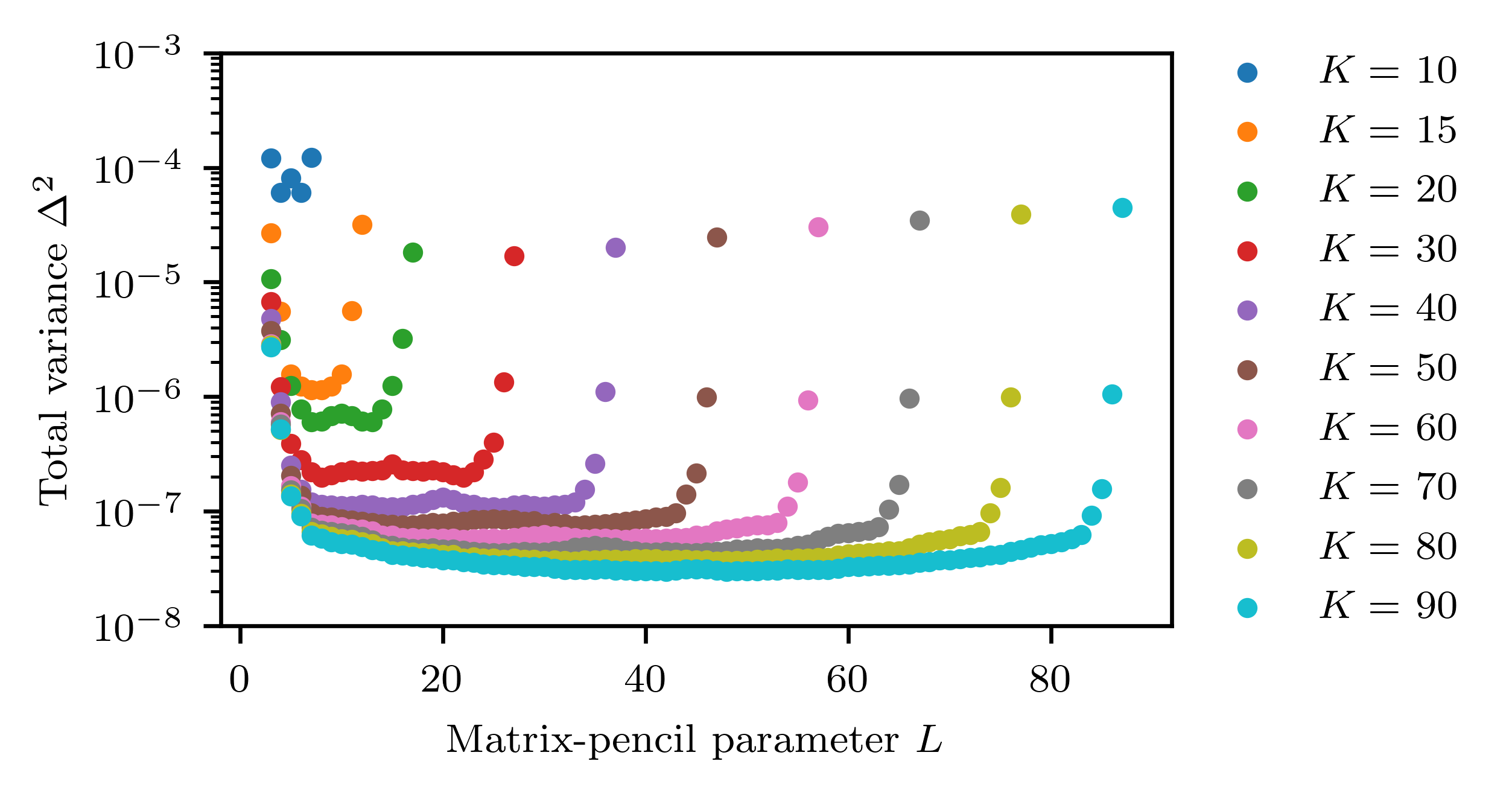}\hspace{5mm}
		\includegraphics[width=0.88\columnwidth]{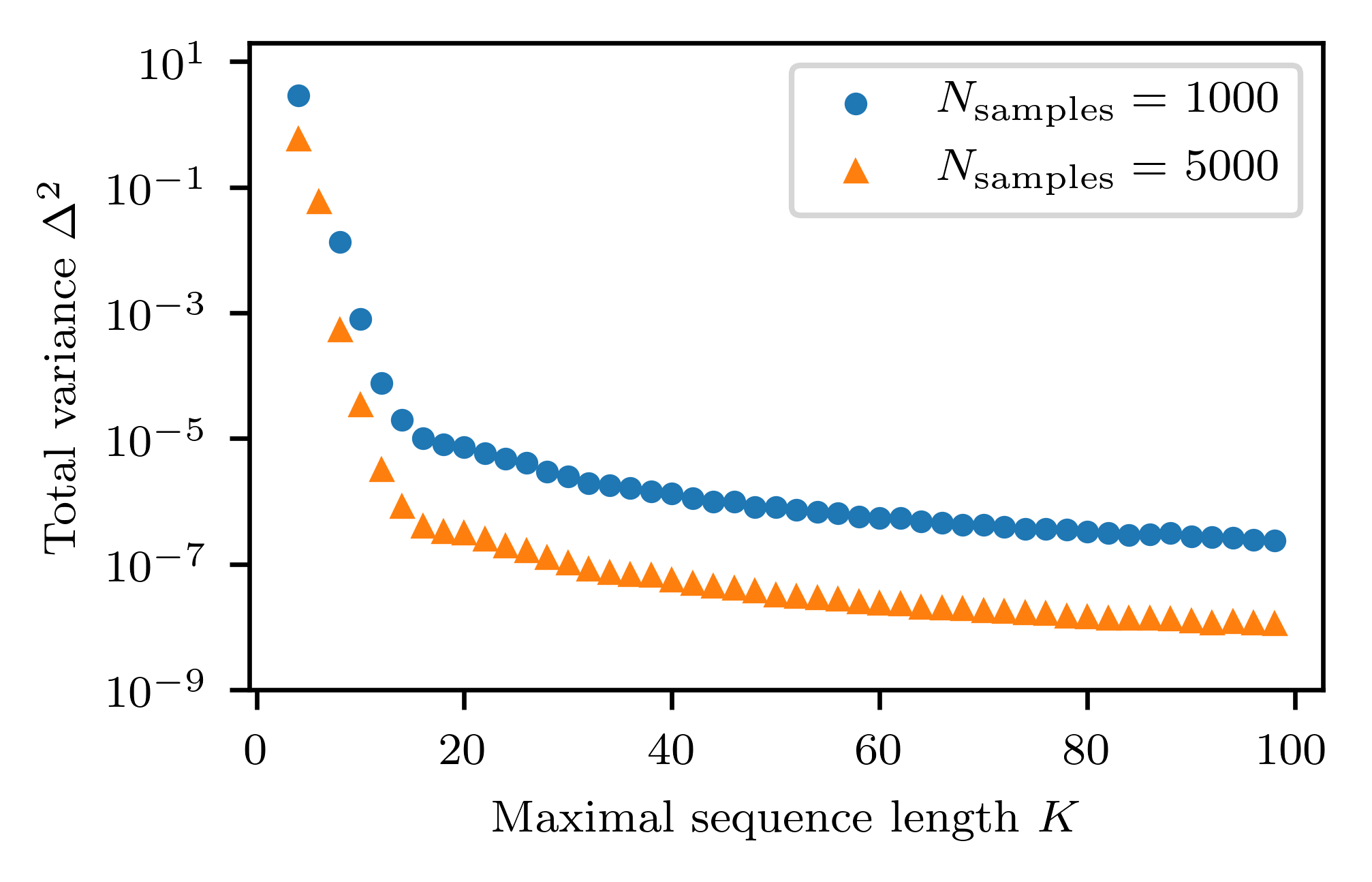}
		\caption{\label{fig:pencil-test}  Preliminary study of the numerical accuracy of the matrix-pencil method as a function of $L$, $K$ and $N_{\rm samples}$. {\bf (Left)} We use the matrix-pencil method with different $L$'s and $K$'s to estimate the eigenvalues of a random single-qubit channel, for $N_{\rm samples}=1000$. On the vertical axis we give the variance in the estimate of the eigenvalues: $\Delta^2=\frac{1}{3}(\sum_{j=1}^{N=3} |\lambda_j-\lambda_j^{\rm est}|^2)$. We see that, as long as the matrix-pencil parameter $L$ is chosen away from $0$ or $K$, the accuracy of the reconstructed signal is nearly independent of $L$. Furthermore, we see that higher $K$'s can achieve a lower $\Delta^2$. {\bf (Right)} We generate a random single-qubit channel and set $L=K/2$. We plot $\Delta^2$ as a function of $K$  for two different values of $N_{\rm samples}=1000$ and $N_{\rm samples}=5000$, showing how a larger $N_{\rm samples}$ suppresses the total variance. We see that for constant $N_{\rm samples}$ the accuracy of the method increases rapidly at first when $K$ is increased, but it increases more slowly if $K$ is already large. This can be explained by the fact that the signal decreases exponentially in $K$ and so data points for large $K$ have much lower signal-to-noise ratio. For both figures, random channels were generated using QuTip's random TPCP map functionality, and measurement noise was approximated by additive Gaussian noise with standard deviation equal to $1/\sqrt{N_{\rm samples}}$. }
		\label{fig:varyL}
	\end{figure*}

	We require at least $K\geq 2N-2$ in order to determine the eigenvalues accurately. This implies that for a single-qubit gate with $N=3$ we need at least $K=4$ and for a two-qubit gate with $N=15$ we need at least $K= 28$. However, the signal $g(k)$ has sampling noise due to a bounded $N_{\rm samples}$ and in practice it is good to choose $K$ larger than strictly necessary to make the reconstruction more robust against noise. We study the effect of varying $K$ in Fig.~\ref{fig:varyL} (left panel).
	
	The method goes as follows and relies on picking a so-called pencil parameter $L$.

Let us assume for now that each $g(k)$ is learned without sampling noise. One constructs a 
	$(K-L+1) \times (L+1)$-dimensional data matrix $Y$ as 
	\begin{align}
	Y&=\left(\begin{array}{cccc} g(0) & g(1) & \ldots & g(L) \\
	g(1) & g(2) & \ldots & g(L+1) \\
	g(2) & \vdots && \vdots \\
	\vdots &  & & \vdots \\
	g(K-L) &\ldots & \ldots & g(K)
	\end{array}\right) \nonumber \\
	&=  \sum_{j=1}^N A_j 
	\left(\begin{array}{cccc} 1 & \lambda_j & \ldots & \lambda_j^L \\
	\lambda_j & \lambda_j^2& \ldots & \lambda_j^{L+1} \\
	\lambda_j^2 & \vdots && \vdots \\
	\vdots & & & \vdots \\
	\lambda_j^{K-L} &\ldots & \ldots & \lambda_j^K
	\end{array}\right).
	\end{align}
	Note that ${\rm rank}(Y)\leq N$ since $Y$ is a sum of at most $N$ rank-1 matrices when there are $N$ eigenvalues. Consider two submatrices of $Y$: the matrix $G_0$ is obtained from $Y$ by deleting the last column of $Y$, while the matrix $G_1$ is obtained by deleting the first column of $Y$. When $L=\frac{K}{2}$, the matrices $G_0$ and $G_1$ are square matrices of dimension $M=\frac{K}{2}+1$. For this choice of $L$, the smallest value of $K$ so that $M=N$ is $2N-2$. We seek a time-shift matrix $\mathfrak{T}$ such that $\mathfrak{T}G_0=G_1$. When $M\geq N$, there certainly exists a matrix $\mathfrak{T}$ such that for all $j\in\{1,\ldots,N\}$:
	\begin{equation}
	\mathfrak{T}\left(\begin{array}{c} 1\\ \lambda_j \\ \vdots \\ \lambda_j^M \end{array}\right)=\lambda_j \left(\begin{array}{c} 1 \\ \lambda_j \\ \vdots \\ \lambda_j^M \end{array}\right).
	\end{equation}
	Furthermore, if $G_0^{-1}$ exists, which is the case when ${\rm rank}(G_0)=M$, this matrix $\mathfrak{T}$ will be uniquely given as $G_1 G_0^{-1}$. Hence, in this case there is a unique matrix $\mathfrak{T}$, obtained by constructing $G_1 G_0^{-1}$ from the data, which is guaranteed to have $\{\lambda_j\}$ as eigenvalues. When the pencil parameter $L > \frac{K}{2}$, one needs to ensure that there are at least $N$ rows of the matrix $Y$: if not, $\mathfrak{T}$ would be of dimension less than $N$, not giving $N$ eigenvalues. This implies $K \geq N+L-1$. 
	
	The general method for a non-square $Y$ which includes an additional sampling-noise reduction step then goes as follows. The choice for $N$ in the procedure can be varied from its minimal value equal to $d^2-1$ to a larger value, depending on a goodness-of-fit.
	\begin{enumerate}
		\item Construct a singular-value decomposition of the matrix $Y$, i.e.~$Y=R_1 \Sigma R_2^T$ and replace the diagonal matrix $\Sigma$ by a diagonal matrix $\Sigma_{\rm clean}$ with only the largest $N$ singular values. 
Let $Y_{\rm clean}=R_1 \Sigma_{\rm clean} R_2^T$. This step is to reduce sampling noise.
		\item Take the submatrices $G_0$ and $G_1$ of $Y_{\rm clean}$.
		\item Compute $\mathfrak{T}=G_1 G_0^+$ where $G_0^+$ is the Moore-Penrose pseudo-inverse of the matrix $G_0$ so that $\mathfrak{T}$ is a matrix with at most $N$ non-zero eigenvalues.
		\item Compute the eigenvalues of $\mathfrak{T}$: these will be the estimates $\lambda_j^{\rm est}$ of $\lambda_j$ for all $j \in \{1, \ldots N\}$. Formally, the linear matrix pencil is $G_0-\lambda G_1$ and the eigenvalues of this matrix pencil, i.e.~the values where ${\rm det}(G_0-\lambda G_1)=0$, are the $\lambda_j^{\rm est}$. 
	\end{enumerate}

We have first applied this method on the signal $g(k)$ of a randomly chosen single-qubit channel: by varying $K$ and $L$ we want to understand the role of the matrix-pencil parameter $L$ and the choice for a larger $K$.
The results are shown in Fig.~\ref{fig:varyL} (left panel). Note that the chosen $K$'s are quite far above the bound $K \geq N+L-1$ to effectively suppress sampling noise. For each $K$ there is a flat region in $L$ where $\Delta^2$ is roughly constant. In the remainder we will choose $L=K/2$, putting ourselves in the middle of this region. Fig.~\ref{fig:varyL} (right panel) shows how increasing $N_{\rm samples}$ lowers the total variance of the estimated eigenvalues.

An additional processing step is the determination of the (complex) amplitudes $\{A_j\}$. Viewing $g(k)$ as a set of $K+1$ inner products between the vector $(A_1, \ldots, A_N)$ and the linearly-independent vectors $(\lambda_1^k, \lambda_2^k,\ldots, \lambda_N^k)$, it is clear that, given perfect knowledge of $g(k)$, the $\{A_j\}$ are uniquely determined when $K+1 \geq N$. Since $g(k)$ is known with sampling noise, the $\{A_j\}$ can be found by solving the least-squares minimization problem $\min_{A_j} \sum_k |g(k)-\sum_j A_j (\lambda_j^{\rm est})^k|^2$. The optimal values in this minimization $A_j^{\rm est}$ and $\lambda_j^{\rm est}$ together form the reconstructed signal $g^{\rm est}(k)$ and the error is given by 
\begin{equation}
\epsilon^{\rm rms}_N=\left(\frac{1}{K+1} \sum_{k=0}^K |g(k)-g^{\rm est}_N(k)|^2\right)^{1/2}.
\label{eq:rms}
\end{equation}

\subsubsection{Resources}
\label{sec:resources}
It is interesting to consider the amount of experiments that must be done to perform spectral quantum tomography. One must estimate the function $g(k)$ defined in \cref{eq:defgk}. 
This reconstruction process requires running $2^n \times N \times (K+1)$ different experiments and repeating each experiment $N_{\rm samples}$ times. For a single-qubit gate we need $6 (K+1)$ experiments, while for a two-qubit gate we need $60 (K+1)$, see Sec.~\ref{sec:resources} for a comparison with randomized benchmarking. Note that while the number of experiments scales exponentially with qubit number (not surprising for a tomographic protocol), the number of experiments needed for performing spectral tomography on single and two-qubit gates is comparable to the number of experiments that must be performed in randomized benchmarking on one or two qubits (which provides only average gate information). In randomized benchmarking one must sample $M$ random sequences for each sequence length $k\in [0:K]$, yielding $M\times (K+1)$ experiments. This $M$ is independent of the number of qubits~\cite{helsen2017multi}. In experiments $M$ is often chosen between $M\approx 40$~\cite{barends2014superconducting,xue2019benchmarking} at the low end and $M\approx 150$ at the higher end~\cite{ballance2016high}. Values of $K$ reported in randomized benchmarking experiments are also comparable to (or even higher than, see~\cite{barends2014superconducting} where $K\approx 300$ is considered) the values of $K$ used for single and two qubit spectral tomography (see \cref{sec:num}).

	\subsection{Spectral tomography on two superconducting chips}
	\label{sec:num}

	\begin{figure*}[t]
		\includegraphics[width=1\columnwidth]{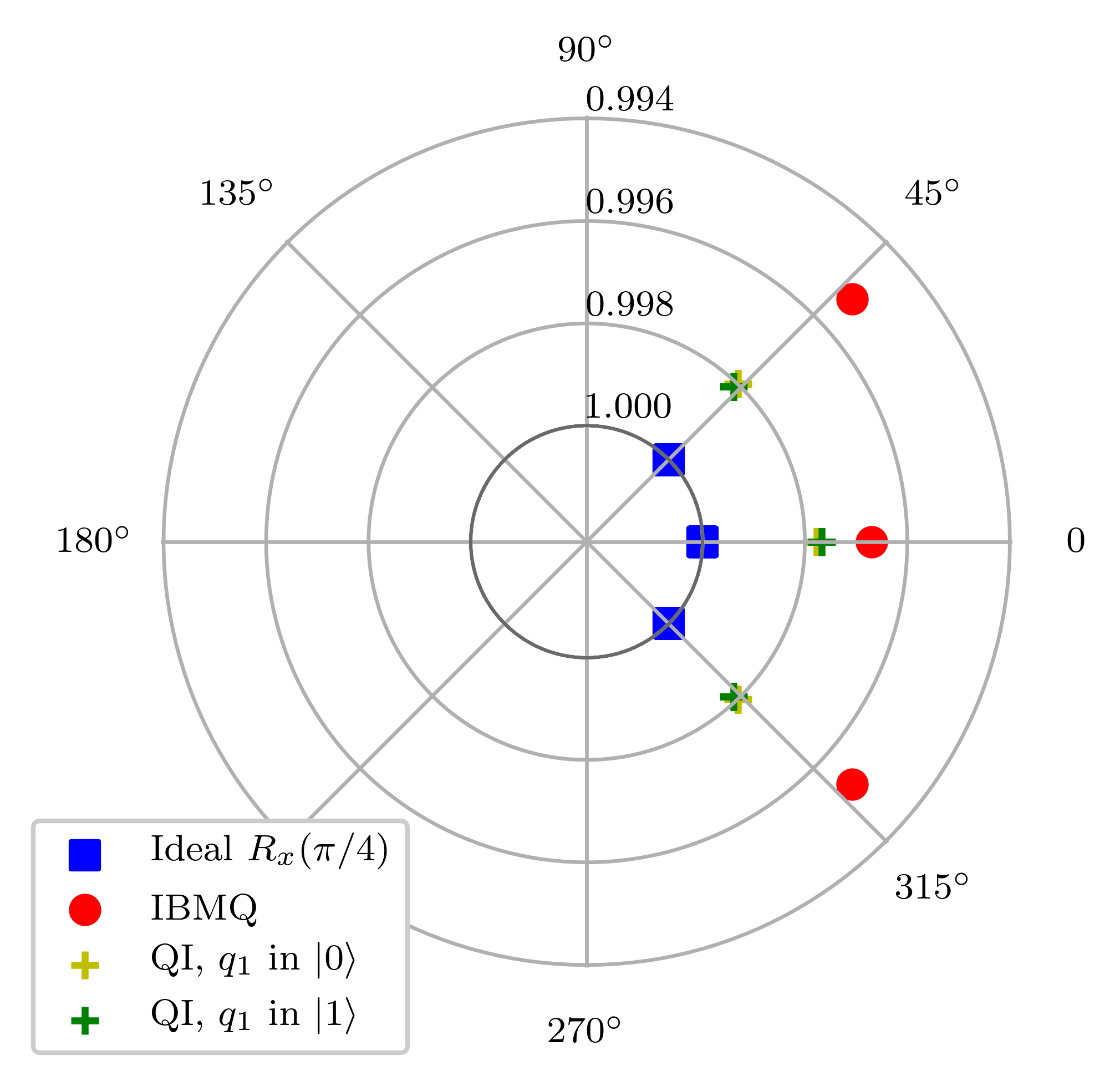}\hspace{1em}
		\includegraphics[width=1\columnwidth]{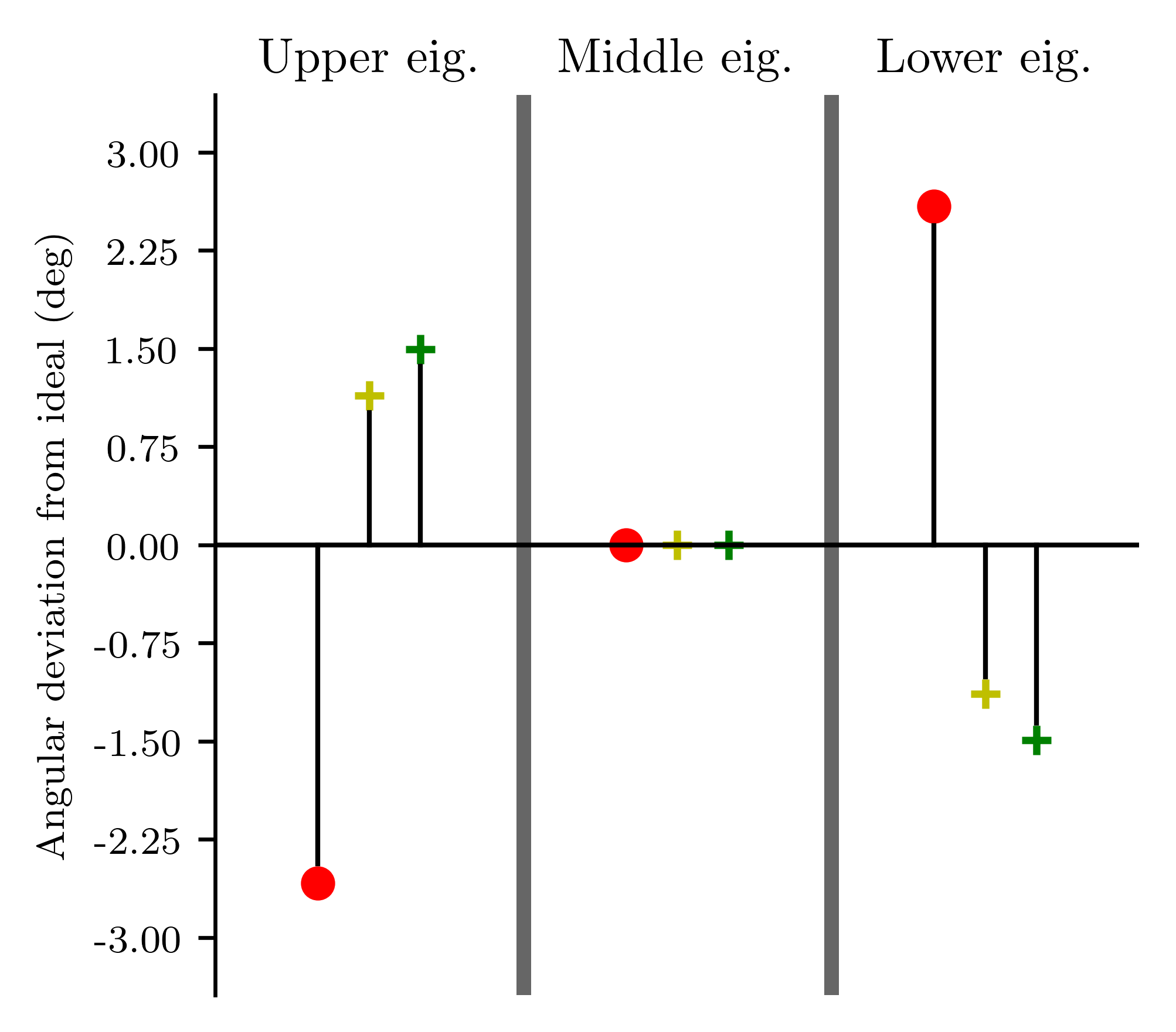}
		\caption{\label{fig:poleplots}{\bf(Left)} Spectral footprints for single-qubit $R_x(\pi/4)$ gates on the {\it ibmqx4} (IBMQ) and the Quantum Infinity (QI) chips at $K=50$, $L=30$ and $N_{\rm samples}=8192$. The modulus of the eigenvalues is plotted in the radial direction and in particular it decreases from the center to the outside and it is equal to 1 on the (most inner) black circumference. The angular coordinate corresponds to the phase of the eigenvalues. {\bf(Right)} Precise value of the deviation of the phases of the three eigenvalues from the ideal ones.}
		\label{fig:pi4}
	\end{figure*}
	
	We have executed the spectral tomography method on a single-qubit $\pi/4$ rotation around the $X$-axis: $R_x(\pi/4)=\exp(-i \pi X/8)$. For this gate the ideal matrix $T^{R_x(\pi/4)}$ should have eigenvalues $1$, $\exp(i \pi/4)$ and $\exp(-i \pi/4)$.
	We execute this gate on two different systems available in the cloud: the two-qubit Quantum Infinity provided by the DiCarlo group at QuTech (for internal QuTech use) and the \emph{ibmqx4} (IBM Q5 Tenerife) available at \url{https://quantumexperience.ng.bluemix.net/qx/editor}. The results of this experiment are shown in Fig.~\ref{fig:pi4} (left panel) in a polar plot which we refer to as the `spectral footprint' of the gate. For clarity, in Fig.~\ref{fig:pi4} (right panel) we have also plotted the phase deviation from ideal for the implemented gates.
	
	On the two-qubit ($q_0$ and $q_1$) Quantum Infinity chip, we perform the single-qubit gate experiment  on $q_0$ twice to study cross-talk: in one case the undriven qubit $q_1$ on the chip is in state $\ket{0}$, in the other case $q_1$ is in state $\ket{1}$. Since the residual off-resonant qubit coupling, mediated through a common resonator, is non-zero, we observe a small difference between these two scenarios, see Fig.~\ref{fig:poleplots}. For the Quantum Infinity chip, when $q_1$ is $\ket{0}$ we estimate $\lambda_j^{\rm est}\in \{0.691+0.719i, 0.691-0.719i, 0.997\}$, while $\lambda_j^{\rm est}\in \{0.687+0.7239i, 0.687-0.724i, 0.998\}$ when $q_1$ is $\ket{1}$. Using the single-qubit fidelity bound given in~\cref{sec:fidel}, we can compute that the fidelity with respect to the targeted gate $R_x(\pi/4)$ can be no more than $0.999$ regardless of the state of $q_1$. We can also compute upper and lower bounds on the unitarity (see \cref{sec:eig} and \cref{sec:fidel}) which yields $0.994\leq u\leq 0.996$ regardless of the state of $q_1$.

	Regarding the {\it ibmqx4} chip, the data are taken when all other qubits are in state $\ket{0}$. The reconstructed eigenvalues $\lambda_j^{\rm est}\in  \{0.735+0.671i, 0.735-0.671i, 0.996\}$ turn out to be lower in magnitude. From these numbers we can conclude that the fidelity to the target gate is no higher than $0.998$ and the unitarity lies between $0.988$ and $0.991$.

	For all these numbers a two-way $95\%$ confidence interval (for both real and imaginary parts) deviates by less than $0.005$ from the quoted values. The confidence intervals are obtained through a Wild resampling bootstrap with Gaussian kernel~\cite{wu1986jackknife}. 
	
	We have considered whether the data can be better fitted with more than $N = 3$ eigenvalues. For each experiment we fit the data using $N$ eigenvalues with $N \in \{4, \ldots 15\}$ and we test whether there is a significant increase in goodness-of-fit using a standard F-test~\cite[Section 2.1.5]{sebernonlinear}. For no experiment and value of $N$ does the resultant $p$-value drop below $0.05$, leading us to conclude that increasing the number of eigenvalues does not significantly increase the accuracy of the fit.

	We have also executed a two-qubit CNOT gate on {\it ibmqx4}~(\cref{fig:CZ-quantum-exp}). The $T$ matrix of the ideal CNOT gate has 15 eigenvalues and a very degenerate spectrum: $6$ eigenvalues are equal to $-1$ and $9$ eigenvalues are equal to $1$, but our data, taking $K=50$, shows that a best fit is obtained using $4$ instead of $2$ eigenvalues. Using an F-test shows that the goodness-of-fit is significantly improved using 4 eigenvalues rather than 2 or 3, whereas adding more eigenvalues beyond $4$ does not significantly improve the goodness-of-fit $(p>0.05)$. 

	We have not tried using larger $K$ (which may lead to a resolution of more eigenvalues) since this would break the requirement that our experiments are executed as a single job performed in a short amount of time on the IBM Quantum Experience. The eigenvalues are $\lambda_j^{\rm est}\in  \{0.939+0.059i,  0.938-0.059i , -0.961+0.067i, -0.961-0.067i\}$, all with a $95\%$ confidence interval smaller than $\pm3\times 10^{-3} $ for both real and imaginary parts.
It is important to note that these 4 eigenvalues, coming in 2 complex-conjugate pairs, {\em cannot be the spectrum} of a two-qubit TPCP map $\mathcal{S}$, for the following reasons. As observed in~\cref{sec:eig}, the submatrix $T^{\mathcal{S}}$ of the Pauli transfer matrix of $\mathcal{S}$ is a real matrix of odd $(4^2-1 =15)$ dimension. Since any complex eigenvalues of a real matrix come in conjugate pairs, $T^{\mathcal{S}}$ must have at least one real eigenvalue. Moreover, the data cannot be explained by allowing for leakage, as any eigenvalues associated to a small amount of leakage must have small associated amplitude, as we discuss in~\cref{sec:nm}. This is not the case for the eigenvalues plotted in~\cref{fig:CZ-quantum-exp} as all their amplitudes have comparable magnitude $A^{\rm est}\in\{ 3.34-1.70i,  3.34+1.70i,  1.57+0.91i,  1.57-0.91i\}$.

 In~\cref{sec:frame} we propose a simple model based on a frame mismatch accumulation that qualitatively reproduces these eigenvalues. This model is not stochastic but coherent, and it violates the assumption that the applied CNOT gate can be fully modeled as a TPCP map. A possible physical mechanism producing a frame mismatch accumulation can be a drift in an experimental parameter.
	
	We do not compute bounds on the fidelity or unitarity of the CNOT gate since the bounds in~\cref{sec:gate_quality_bounds} do not apply when the evolution is non-Markovian.
	
	\begin{figure}[htb]
		\includegraphics[width=1.0\columnwidth]{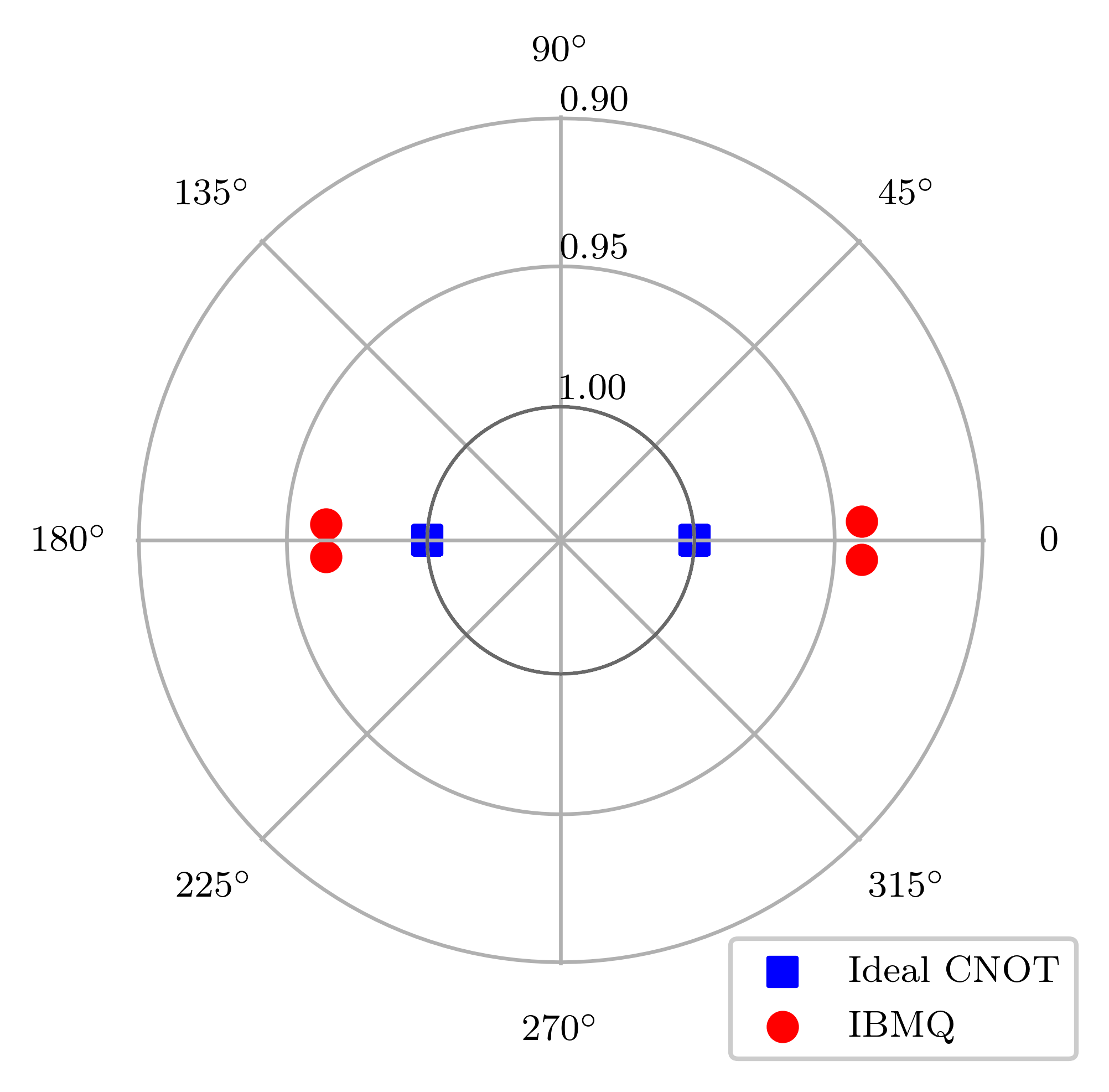}
		\caption{\label{fig:CZ-quantum-exp}Spectral footprint of the CNOT gate for $K=50$ and $N_{\rm samples}=8192$. Even though the CNOT gate has only two (degenerate) eigenvalues, we find that the spectrum of the noisy gate can be best described using $4$ distinct eigenvalues. The fact that none of them are real suggests that the data cannot be due to the repeated execution of the same noisy gate. In~\cref{sec:frame} we propose a simple coherent non-Markovian model that offers a possible mechanism for the absence of real eigenvalues.}
	\end{figure} 
\subsection{Leakage and non-Markovian noise}
	\label{sec:nm}
	In this section we consider how spectral tomography behaves under error models that violate the assumptions that go into~\cref{eq:defgk}. 
	
	Three common mechanisms for gate inaccuracy are (1) cross-talk, meaning the gate depends on or affects the state of other ``spectator" qubits, (2) leakage, meaning that the dynamics of the gate acts outside of the computational qubit subspace and (3) non-Markovian dynamics, meaning that the assumption that $k$ applications of the noisy gate are equal to ${\cal S}^k$ for some TPCP map $\mathcal{S}$ is incorrect. Characterizing gates with respect to these features is important for assessing their use in multi-gate/multi-qubit devices for the purpose of quantum error correction or plainly reliable quantum computing~\cite{WG:leak, rol+:netzero}.
	
	One can see that all three scenarios are due to the dynamics taking place in a larger Hilbert space than the targeted computational qubit space. In the case of leakage, the larger space is an extension of the computational space, while in the other cases the larger space is the tensor product of the computational space with the state space of spectator qubits (1), as explored in~\cref{sec:num}, or other quantum or classical degrees of freedom in the environment (3).

	\subsubsection{Leakage}
	
	Let us consider how gate leakage affects the signal $g(k)$, making the analysis for one or two \emph{qutrits}.  One can choose an operator basis for the qutrit space such as the basis of the 8 traceless (normalized) Gell-Mann matrices $\sigma_\mu^{\rm GM}$ for $\mu=1,\ldots, 8$, together with the normalized identity $\sigma_0^{\rm GM}=\frac{1}{\sqrt{3}} I_3$. 
	For a single qutrit, we can consider the `Pauli' transfer matrix in this Gell-Mann basis, i.e.~$S^{\rm GM}_{\mu \nu}={\rm Tr}[\sigma_{\mu}^{\rm GM} {\cal S}(\sigma_{\nu}^{\rm GM})]$ and its submatrix $T^{\rm GM}$. 
	
	For a single qutrit, the signal $g^{\mbox{\tiny NO SPAM}}(k)$ in Eq.~(\ref{eq:nospam}) then equals ${\rm Tr}_{\rm comp}[(T^{\rm GM})^k]$ where ${\rm Tr}_{\rm comp}[A]$ represents the trace over a $3\times 3$ submatrix of $A$, corresponding to the Gell-Mann matrices which act like X, Y, and Z in the two-dimensional computational space.
	In other words, we can see the matrix $T^{\rm GM}$ as being composed of blocks:
	\begin{equation}
	T^{\rm GM}=\left(\begin{array}{cc} T_{\rm comp} & T_{\rm seep} \\
	T_{\rm leak} & T_{\rm beyond}  \end{array}\right),
	\label{eq:T_GM}
	\end{equation}
	where the upper-left block is the sub-matrix whose trace we take in $g^{\mbox{\tiny NO SPAM}}(k)$. In the absence of other noise sources, $T^{\rm GM}$ corresponds to the evolution of a unitary gate and (assuming it is diagonalizable) it can be diagonalized by a rotation $V$ as $T^{\rm GM}=V D V^{-1}$, where $D$ is a diagonal matrix with all the eigenvalues $\{\lambda_j\}$.
	If we assume that leakage is low, meaning that $T_{\rm leak}$ and $T_{\rm seep}$ have small norm of $O(\epsilon)$, then at lowest order in $\epsilon$ the diagonalizing transformation $V$ will be block-diagonal, i.e.~$V \approx V_{\rm comp} \oplus V_{\rm beyond}$. This means that $g^{\mbox{\tiny NO SPAM}}(k)={\rm Tr}_{\rm comp}\big[ (T^{\rm GM})^k\big]= {\rm Tr}_{\rm comp}\big[ V D^k V^{-1}\big] \approx \sum_{j=1}^3 \lambda_j^k+O(\epsilon)$. Thus, at lowest order, the signal will have large amplitude on 3 relevant eigenvalues of the spectrum of $T^{\rm GM}$ and these eigenvalues could have been perturbatively shifted from their ideal location by low leakage. If the leakage is stronger, we can more generally write 
	\begin{equation}
	g^{\mbox{\tiny NO SPAM, LEAK}}(k)=\sum_{j=1}^8 \tilde{A}_j\lambda_j^k,  \;\tilde{A}_j=\bra{\sigma_j} V^{-1} \Pi_{\rm comp} V \ket{\sigma_j}.
	\label{eq:leak}
	\end{equation}
	Here $\ket{\sigma_j}$ is a vector notation for one of the 8 Gell-Mann matrices $\sigma_j$ and $\Pi_{\rm comp}$ is the projector onto the basis spanned by the 3 Gell-Mann matrices which are the Paulis in the computational space. From this expression we see that the effect of leakage is the contribution of more eigenvalues to the signal $g(k)$. For low leakage we may expect three dominant eigenvalues with relatively large amplitude $\tilde{A}_j$ and five eigenvalues with small amplitude.
	
	For a gate on two qutrits, identical remarks apply, except that an additional basis transformation is required from the orthogonal Gell-Mann basis to the computational qubit Pauli basis in order to keep the same division of~$T^{\rm GM}$ as in~\cref{eq:T_GM}. If we have two qutrits, the 80-dimensional traceless subspace is spanned by the matrices $\sigma_\mu^{\rm GM} \otimes \sigma_\nu^{\rm GM}$ for $\mu,\nu=0,\ldots, 8$ except $\mu=\nu= 0$. 
	The issue is related to terms such as $\sigma_0^{\rm GM} \otimes \sigma_{\nu\neq 0}^{\rm GM}$ since the normalization of the qutrit identity ($\sigma_0^{\rm GM}=\frac{1}{\sqrt{3}} I_3$) is different from the normalization of the qubit identity ($P_0=\frac{1}{\sqrt{2}} I_2$).
	This suggests that for two qutrits it is better to write $T^{\rm GM}$ in a basis which includes the Pauli matrices in the computational subspace ($P_{\mu} \otimes P_{\nu}$ for $\mu, \nu=0,\ldots, 3$ except $\mu=\nu = 0$) as a sub-basis. For two qutrits, the signal may then contain up to 80 eigenvalues of which all but 15 are expected to have small amplitude in case of low leakage.

	\subsubsection{Non-Markovianity: time-correlated noise}
	
	Non-Markovian behavior of a gate can be due to temporal correlations in the classical or quantum environment of the driven qubit(s). Abstractly, we can include the environment in the gate action so that the evolution for each gate application is a unitary given by some $U_{\rm total}$ acting on system and environment. We can expand the Pauli transfer matrix of $U_{\rm total}$ in a Pauli basis for system and environment and view $T_{\rm comp}$ as a sub-block of $T_{\rm total}$, similar as in the case of leakage. Diagonalizing $T_{\rm total}$ and taking the trace over the computational space will result in an expression such as Eq.~(\ref{eq:leak}). For example, an additional spectator or environment qubit can lead to a signal $g(k)$ of a single-qubit gate having contributions from 15 eigenvalues. Choosing a sufficiently large $K$ may allow one to resolve these eigenvalues, even those with small amplitude.
	\begin{figure}[htb]
		\includegraphics[width=1.0\columnwidth]{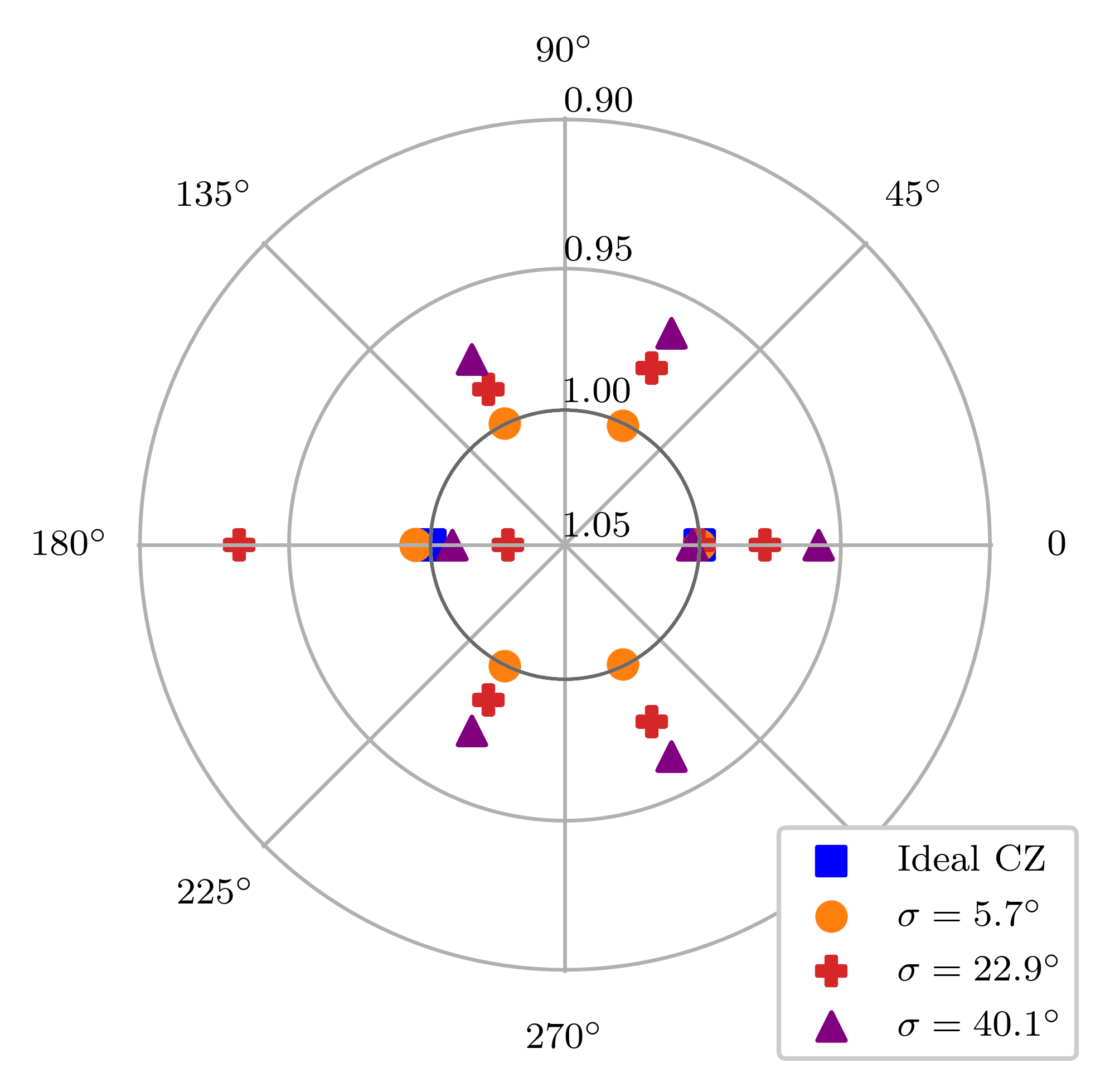}
		\caption{\label{fig:CZ_gaussian} Spectral footprint of a simulated CZ gate affected by non-Markovian noise quantified by $\sigma$, see Eq.~(\ref{eq:CZ_gaussian_integral}). For each $\sigma$ we use an F-test ($p$-value 0.01) to find the number of eigenvalues that best fit the simulated $g^{\mbox{\tiny NO SPAM}}(k)$ with $K=50$. We find respectively 7, 12 and 11 eigenvalues for $\sigma=5.7^\circ, 22.9^\circ, 40.1^\circ$ (here we show only the eigenvalues with modulus greater than 0.9). We observe eigenvalues with modulus larger than 1 if $\sigma$ is sufficiently large. These results are qualitatively stable if we add a small amount of sampling noise.} 
		\label{fig:nonM}
	\end{figure}
	
	A more malicious, but physically reasonable~\cite{rol+:netzero}, form of classical non-Markovian noise makes gate-parameters temporally correlated.
	In order to numerically study the effect of non-Markovian noise, we consider a toy example in which a perfect CZ gate is followed by a rotation around the $X$ axis on one qubit. 
	For a series of $k$ repetitions of a perfect CZ gate, we assume that each one is followed by the same rotation $R_x(\phi)$ acting always on the same qubit. We assume that the angle $\phi$ is Gaussian-distributed with mean 0 and standard deviation $\sigma$: $\mathbb{P}_\sigma(\phi)=\exp(-\phi^2/2\sigma^2)/\sqrt{2\pi}\sigma$.
	The time evolution for $k$ repetitions is then given by
	\begin{equation}
	\mathcal{S}_k(\rho)= \int_{-\infty}^{+\infty} d\phi\, \mathbb{P}_\sigma(\phi) \bigl(R_x(\phi) \,\mathrm{CZ}\bigr)^k \rho \bigl(\mathrm{CZ} \, R_x(\phi)^\dagger \bigr)^k.
	\label{eq:CZ_gaussian_integral}
	\end{equation}
	Note that $\mathcal{S}_k\neq (\mathcal{S}_1)^k$ since this noise is correlated across multiple repetitions of the gate.
	Furthermore, we assume perfect state-preparation and measurement. In this case, one can represent the noisy gate by some unitary $U_{\rm total}$ acting on the two qubits and on a classical state in a Gaussian stochastic mixture of angles $\phi$. The continuous nature of this classical environment state leads to a lack of a hard cut-off on the number of eigenvalues in $g(k)$. 
	
	We apply the matrix-pencil method to the corresponding signal $g^{\mbox{\tiny NO SPAM}}(k)$ and we use an F-test to determine the optimal number of eigenvalues for each~$\sigma$~(\cref{fig:CZ_gaussian}).
	For~$\sigma=22.9^\circ$ and $K=50$ we find eigenvalues with modulus clearly larger than~1.
	Those are unphysical but not excluded by the matrix-pencil method.
	We expect that such $|\lambda^{\rm est}|> 1$ disappear when considering a longer signal, since $g(k)$ does not increase exponentially in $k$. In other words, this is a sign that the signal contains more spectral content than can be resolved from the time scale set by $K$.
	Indeed, for~$\sigma=22.9^\circ$ we have made the same analysis for larger~$K$'s up to $K=200$ and we find that those eigenvalues get closer and closer to~1.
	If instead we fix $K=50$ and consider different~$\sigma$'s, we find that for a low~$\sigma$ (e.g.~$5.7^\circ$) unphysical eigenvalues are not present~(\cref{fig:CZ_gaussian}), whereas for $\sigma>22.9^\circ$ (e.g.~$40.1^\circ$) they get again closer and closer to~1.
	This latter fact can be understood by noting that increasing $\sigma$ is analogous to enlarging the time scale set by~$K$, as the characteristic time scale of dephasing gets shorter for a fixed~$K$.
	Based on these observations, we conclude that there is a certain intermediate time scale at which  eigenvalues larger than~1 are extrapolated from the data in the presence of sufficiently-strong non-Markovian noise of the kind described in this section. \cref{sec:frame} discusses a model with a different kind of time-correlation leading to a spectral footprint which is incompatible with that of a TPCP map.
		
	\subsubsection{Non-Markovianity: coherent revivals}
	\begin{figure}[htb]
		\includegraphics[width=1.0\columnwidth]{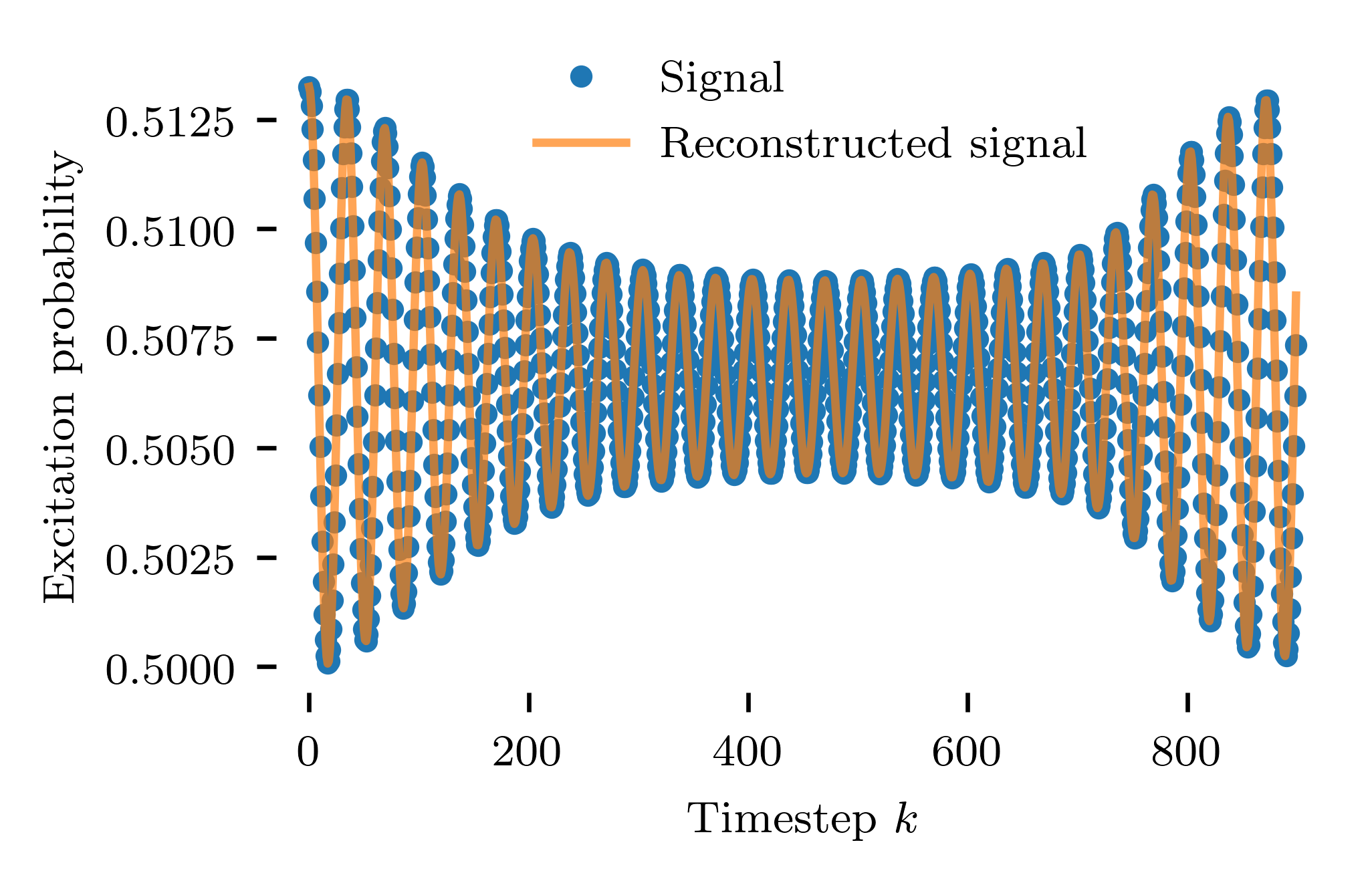}
		\includegraphics[width=1.0\columnwidth]{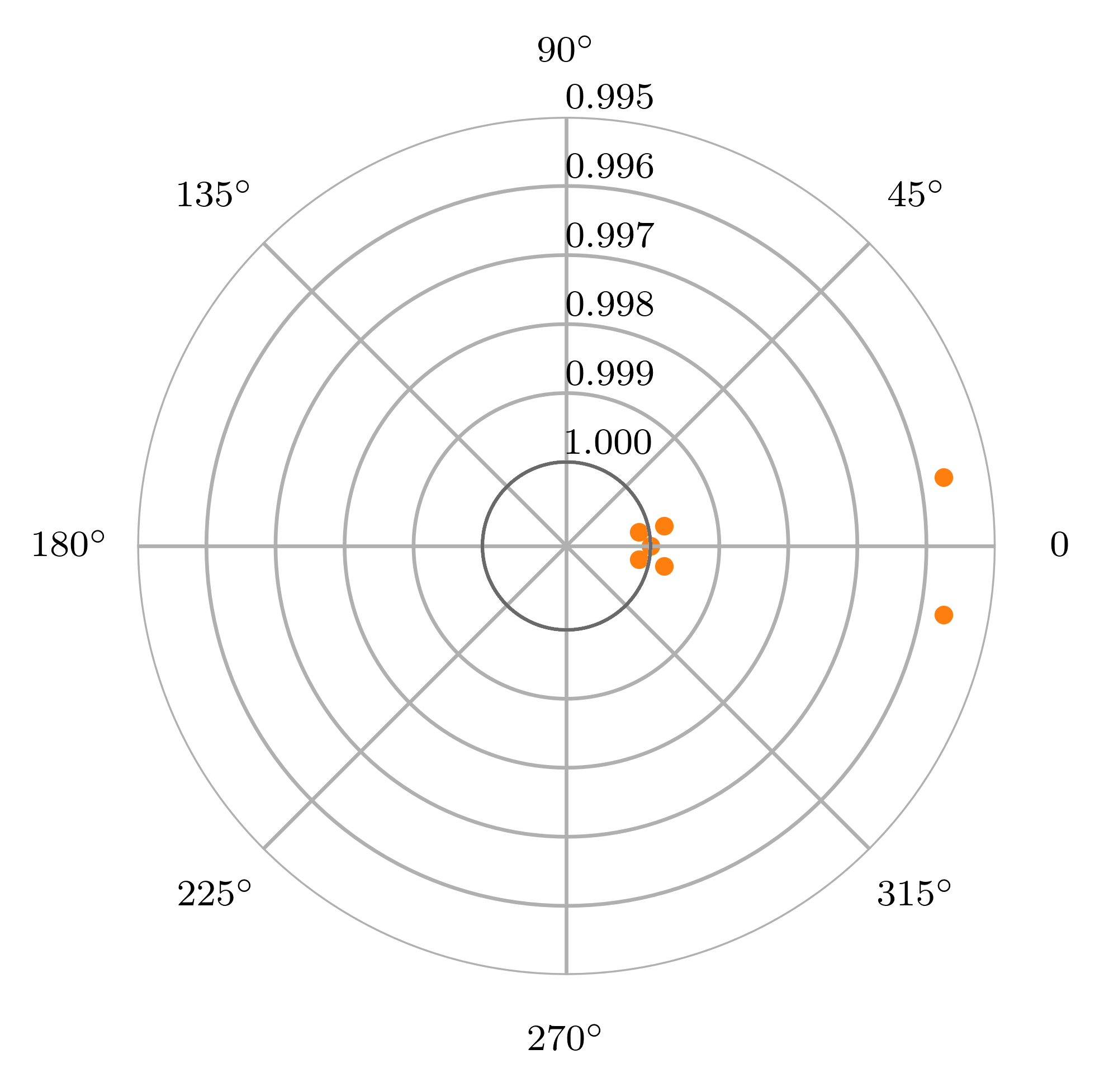}
		\caption{Study of the reviving signal given in~\cref{eq:excitation_probability} for $k\cdot\Omega \delta t=k\cdot 0.05$, $\bar{n}=5$ and $K=900$. We find that the reviving signal is well reconstructed by a fit with 15 eigenvalues, some of which are distinctly separated as can be seen in the spectral footprint. Some of the eigenvalues are estimated to be larger than 1. This is another example in which the matrix-pencil method gives unphysical eigenvalues in the presence of non-Markovian behavior (revivals here, time-correlated parameters in~\cref{fig:nonM}).
			\label{fig:revival}}
	\end{figure}
	
	In order to better understand the occurrence of eigenvalue estimates $|\lambda^{\rm est}| > 1$, we apply the matrix-pencil method on a signal (of a somewhat different physical origin), which has a revival over the time period set by $K$.
	
	It is well-known that in the exchange of energy between a two-level atom with a bosonic mode, the Rabi oscillations of the two-level atom are subject to temporal revivals. These revivals are due to the fact that the bosonic driving field is not purely classical, but rather gets entangled with the state of the qubit via the Jaynes-Cummings interaction. In particular, for a coherent driving field with coherent amplitude $\alpha$ with average photon number $\bar{n}=|\alpha|^2$, the probability for the atom to be excited equals (see Section 3.4.3 in~\cite{book:HR}):
	\begin{equation}
	P_e(t)=\frac{1}{2}+\frac{1}{2}\sum_{n=0}^{\infty} p_{\alpha}(n) \cos (\Omega t \sqrt{n+1}),
	\label{eq:excitation_probability}
	\end{equation}
	with $p_{\alpha}(n)=\exp(-|\alpha|^2) \frac{|\alpha|^{2n}}{n!}$. 
	We consider $\bar{n}=5$ and sample the damped oscillatory function $P_e(t)-\frac{1}{2}$ at regular intervals $k \Omega \delta t$ with $\Omega \delta t=0.05$ and $k=0, \ldots, K=900$. The signal function $g(t)=P_e(t)-\frac{1}{2}$ contains eigenvalues equal to $\lambda_n=\exp(\pm i \, 0.05 \sqrt{n+1})$ with amplitudes according to the Poisson distribution $p_{\alpha}(n)$ with mean photon number $\bar{n}$. 

	We observe that the matrix-pencil method finds eigenvalues larger than $1$, see Fig.~\ref{fig:revival}, which contribute significantly ($p<0.01$ via F-test) to the reconstructed signal. We can understand this feature of eigenvalues exceeding 1 as a way in which the matrix-pencil method handles revivals: the signal has more spectral content than what can be resolved from the window of time given by $K$, in particular there is no hard cut-off on the number of eigenvalues which contribute. We have observed that an analysis of the signal over a longer period of time, that is, a larger $K$ up to $K=5000$, gives eigenvalues whose norm converges to at most 1.
	\section{Discussion}
	We have introduced spectral quantum tomography, a simple method that uses tomographic data of the repeated application of a noisy quantum gate to reconstruct the spectrum of this quantum gate in a manner resistant to SPAM errors. We have experimentally validated our method on one- and two-qubit gates and have also numerically investigated its behavior in the presence of temporally-correlated non-trivial error models.

The effective upshot of leakage and non-Markovian noise is that the signal will have more spectral content than what can be resolved given a chosen sequence length $K$, leading to unphysical features in the spectrum such as an eigenvalue estimate larger than 1, or the absence of a real eigenvalue. Even though we have seen in our examples that a physical spectrum can be regained by going to larger $K$, depending on the noise model, this convergence may be very slow requiring much data-taking time. Hence these unphysical features are useful markers for deviations from our model of repeated TPCP qubit maps ${\cal S}^k$. We view it as an open question how well one can reliably distinguish different sources of deviations.

\subsection{Logical Spectral Quantum Tomography}
\label{sec:qec}

	An interesting application of the spectral tomography method could be the assessment of logical gates on encoded quantum information in a SPAM-resistant fashion. In this logical scenario (for, say, a single logical qubit), one first prepares the eigenstates of the logical Pauli operators $\overline{X}, \overline{Y}$ and $\overline{Z}$. One then applies a unit of error-correction $k=0,\ldots, K$ times: a single unit could be, say, the repeated error correction for $L$ rounds of a distance-$L$ surface code. Or a unit is the application of a fault-tolerant logical gate, e.g. by means of code-deformed error correction or a transversal logical gate followed by a unit of error correction. After $k$ units one measures the logical Pauli operators fault-tolerantly, and repeats experiments to obtain the logical signal $\overline{g}(k)$. Studying the spectral features of such logical channel will give information about the efficacy of the quantum error correction unit and/or the applied logical gate while
departures from the code space or a need to time-correlate syndrome data beyond the given QEC unit  
can show up as leakage and non-Markovian errors.

\section{Methods}
	
	\subsection{Single-qubit case with non-diagonalizable matrix $T$}
	\label{sec:nondiag}
	
	In general, a matrix $T$ can be brought to Jordan normal form by a similarity transformation, i.e.~$T=V J V^{-1}$ with $J=\oplus_i J_i$ where each Jordan block $J_i$ is of the form
	\begin{equation}
	J_i=\left(\begin{array}{ccccc} 
	\lambda_i & 1 & & \\ 
	& \lambda_i & \ddots &  \\
	& & \ddots & 1 \\
	& & & \lambda_i 
	\end{array}\right),
	\end{equation}
	see e.g.~Theorem 3.1.11 in~\cite{book:HJ}. $T$ is diagonalizable when each Jordan block is fully diagonal.

	An example of a non-diagonalizable Lindblad superoperator on a single qubit has been constructed in 
	\cite{SL:nondiag}. Using this, one can easily get a single-qubit superoperator ${\cal S}$ for which the traceless block of the Pauli transfer matrix is a non-diagonalizable matrix $T$ as follows. Let ${\cal S}(\rho)=\exp({\cal L}\epsilon)(\rho)\approx \rho+\epsilon {\cal L}(\rho)+O(\epsilon^2)$ with ${\cal L}(\rho)=-i [\frac{yZ}{2}, \rho]+{\cal D}[(2x)^{1/2} \sigma_-](\rho)+{\cal D}[y^{1/2}X](\rho)$ with ${\cal D}[A](\rho)=A \rho A^{\dagger}-\frac{1}{2}\{A^{\dagger} A,\rho\}$ and real parameters $x,y\geq 0$.
	This implies that ${\cal S}$ has the $4 \times 4$ Pauli transfer matrix 
	\begin{equation*}
	S= \left(\begin{array}{cccc} 1 & 0 & 0 & 0 \\ 0 & 1-\epsilon x & -\epsilon y & 0 \\ 0 & \epsilon y & 1-\epsilon (x+2y) & 0 \\ 2\epsilon x & 0 & 0 & 1-2\epsilon (x+y) \end{array}\right)+O(\epsilon^2).
	\end{equation*}
	Taking some small $\epsilon$ and $x \neq 0$, one can check that the submatrix $T$ does not have 3 eigenvectors and it has a pair of degenerate eigenvalues, so $T$ is not diagonalizable. When we take $x=0$, ${\cal S}$ is unital, that is ${\cal S}(I)=I$, and the submatrix $T$ is not diagonalizable either.

	Even though a matrix $T$ is not always diagonalizable, there still exists the so-called Schur triangular form for any matrix $T$~\cite{book:HJ}. This form says that $T=W (D+E) W^{\dagger}$ with $W$ a unitary matrix, $D$ a diagonal matrix with the eigenvalues of $T$, and $E$ a strictly upper-triangular ``nilpotent" matrix with non-zero entries only above the diagonal. Since the $N \times N$ matrix $E$ is strictly upper-triangular, one has ${\rm Tr}\big[D^iE^j\big]=0$ for all $j\neq 0$. If we use this form in Eq.~(\ref{eq:nospam}), one obtains for any $k$
	\begin{equation}
	g^{\mbox{\tiny NO SPAM}}(k)={\rm Tr}\big[T^k\big]={\rm Tr} \big[(D+E)^k\big]={\rm Tr}\big[D^k\big],
	\label{eq:exp-decay}
	\end{equation}
	since any product of the form $D^{l_1} E^{l_2} D^{l_3} \ldots E^{l_m}$ with some non-zero $l_i > 0$ is a matrix with zeros on the diagonal. In case of SPAM errors and non-diagonalizable $T$ we consider 
	\begin{equation}
	g(k)={\rm Tr}\big[W^{\dagger} T_{\rm prep} T_{\rm meas} W (D+E)^k\big],
	\end{equation}
	where $W^{\dagger} T_{\rm prep} T_{\rm meas} W$ is not the identity matrix due SPAM errors, implying that $g(k)$ can depend on $E$ and have a non-exponential dependence on $k$. Thus, in the special case of a non-diagonalizable matrix $T$, the signal $g(k)$ would not have the dependence on the eigenvalues as in Eq.~(\ref{eq:defgk}).
	
	In particular, we can examine the physically-interesting non-diagonalizable \emph{Case 3} in~\cref{sec:lind} in this light, taking $h_y=h_z=0$ and a critical $\hcritical=\frac{\Gamma_1-\Gamma_2}{2}$. The dynamics of the Lindblad equation after time $t$ induces a superoperator ${\cal S}_t$ which will have the following action on the Pauli operators:
	\begin{eqnarray*}
		{\cal S}_t(X)& =& \exp(-\Gamma_2 t) X, \\
		{\cal S}_t(Y)& =& \exp(-(\Gamma_1+\Gamma_2)t/2)\left[(1+t \hcritical) Y-\hcritical Z\right],\\
		{\cal S}_t(Z)& =& \exp(-(\Gamma_1+\Gamma_2)t/2)\left[\hcritical t Y+(1-\hcritical  t)Z\right].
	\end{eqnarray*}
	Here we can note the linear dependence on $t$ due to the system being critically damped. If we consider the signal $g(t)=\sum_{\mu} {\rm Tr}[P_{\mu} {\cal S}_t(P_{\mu})]$ we see that this linear dependence on $t$ drops out in accordance with Eq.~(\ref{eq:exp-decay}), i.e.~this trace only depends on the eigenvalues and has an exponential dependence on $t$.  In the presence of SPAM errors, some of the linear dependence could still be observable for such critically-damped system. In addition, coefficients such as $c_{\mu \nu}(t)={\rm Tr}[P_{\mu} {\cal S}_t (P_{\nu})]$ can depend linearly on $t$, making such tomographic data less suitable to extract eigenvalue information.

	\subsection{Upper bound on the entanglement fidelity with the targeted gate}
	\label{sec:fidel}

	In this section we show how to relate the eigenvalues of the Pauli transfer matrix of a TPCP  map ${\cal S}$ to an upper bound on the entanglement fidelity (and hence the average gate fidelity) with the targeted unitary gate $U$. Naturally, one can only expect to obtain an upper bound on the gate fidelity, since the eigenvalues do not provide information about the eigenvectors of ${\cal S}$. If the actual eigenvectors deviate a lot from ideal, the actual gate fidelity could be very low, so one can certainly not derive a lower bound on the fidelity based on the eigenvalues.

	\begin{lemma}
		\label{lem:upperbound_fidelity}
		Let the eigenvalues of the $N \times N$ matrix $T^{\cal S}$ be $\{\lambda_i\}_{i=1}^N$ with $N=d^2-1$ for a $d$-dimensional system. Let $U$ be the targeted gate with eigenvalues $\{\lambda_i^{\rm ideal}\}_{i=1}^N$ and let there be permutation $\pi$ of $i$-th eigenvalue $\lambda_i$ which maximizes $|\sum_i \lambda_{\pi(i)}^*\lambda_i^{\rm ideal}|$ so that $0 \leq \xi_{\rm max}=\max_{\pi}\frac{1}{N}|\sum_i \lambda_{\pi(i)}^*\lambda_i^{\rm ideal}|\leq 1$. The entanglement fidelity ${\cal F}_{\rm ent}(U, {\cal S})=\frac{1}{N+1}(1+ {\rm Tr}\big[ T^{U^{\dagger}} T^{\cal S}\big])$ is upper bounded as 
		\begin{equation}\label{eq:fid_bound}
		{\cal F}_{\rm ent}(U, {\cal S})\leq \frac{1}{N+1}\left(1+ N
		\sqrt{u(\mathcal{S})-\frac{\sum_j |\lambda_j|^2}{N}}+N\xi_{\rm max}\right),
		\end{equation}
		where $u(\mathcal{S})$ is the unitarity of $\mathcal{S}$.
	\end{lemma}
	
	\begin{proof}
		We write $T^{\cal S}$ in Schur triangular form as $T^{{\cal S}}=W (D^{\cal S}+E) W^{\dagger}$ with $W$ a unitary matrix, $D^{\cal S}$ a diagonal matrix with the eigenvalues of $T^{\cal S}$, and $E$ a strictly upper-triangular ``error" matrix with  non-zero entries only above the diagonal~\cite{book:HJ}.  Using the Cauchy-Schwartz inequality one has 
		\begin{eqnarray}
		{\rm Tr}\big[{T^U}^{\dagger} T^{\cal S}\big]\leq {\rm Tr}\big[{T^U}^{\dagger} W D^{\cal S} W^{\dagger}\big]+ \nonumber \\
		({\rm Tr}\big[E^{\dagger}E\big])^{1/2} ({\rm Tr}\big[{T^U}^{\dagger} T^U\big])^{1/2}. 
		\label{eq:bound}
		\end{eqnarray}
		Note that for a unitary gate $U$, $T^{U^\dagger}=(T^{U})^T=(T^{U})^\dagger$ and $T^{U^{\dagger}} T^U=I$ implying that $T$ is an orthogonal matrix with unit singular values. We thus have $({\rm Tr}\big[{T^U}^{\dagger} T^U\big])^{1/2}=\sqrt{N}$.  One has ${\rm Tr}\big[{T^{\cal S}}^{\dagger} T^{\cal S}\big]={\rm Tr}\big[(D^{\cal S}+E)^{\dagger} (D^{\cal S}+E)\big]={\rm Tr}\big[({D^{\cal S}}^{\dagger} D^{\cal S}+E^{\dagger}E)\big]$, using the strict upper-triangularity of $E$. In other words, ${\rm Tr}\big[E^{\dagger} E\big] ={\rm Tr}\big[{T^{\cal S}}^{\dagger} T^{\cal S}\big]-\sum_i |\lambda_i|^2$ where $\lambda_i$ are the eigenvalues of $T^{\cal S}$.

		Recognizing that $\frac{1}{N}{\rm Tr}\big[{T^{\cal S}}^{\dagger} T^{\cal S}\big] =u(\mathcal{S})$,  we obtain an upper bound on the second term in Eq.~(\ref{eq:bound}). 
		
		Now let's upper bound the first term in Eq.~(\ref{eq:bound}) for unknown unitary $W$. 
		Assume w.l.o.g.~that $T^U$ and $D^{\cal S}$ are diagonal in the same basis (the additional rotation between these eigenbases can be absorbed into $W$). Let $T^U=\sum_i \lambda_i^{\rm ideal} P_i$ and $D^{\cal S}=\sum_i \lambda_i P_i$ with orthogonal projectors $P_i$ and $\sum_i P_i=I$.
		Define the matrix $M$ with entries $M_{ij}={\rm Tr}\big[P_i W P_j W^{\dagger}\big]$.  The matrix $M$ is doubly-stochastic, since $\sum_i M_{ij}=1=\sum_j M_{ij}$ which implies that $M=\sum_m q_m \pi_m$ with $q_m \geq 0, \sum_m q_m=1$ (Birkhoff-von Neumann theorem~\cite{book:HJ}) with permutation matrix $\pi_m$. With these facts and the convention $\bra{i} \lambda^{\cal S}\rangle=\lambda_i$ we can bound
		\begin{eqnarray*}
			|{\rm Tr}\big[{T^U}^{\dagger} W D^{\cal S} W^{\dagger}\big]| \leq \sum_m q_m |\bra{\lambda^{\rm ideal}} \pi_m \ket{\lambda^{{\cal S}}}| \leq N \xi_{\rm max}.
		\end{eqnarray*}
These bounds together then lead to Eq.~(\ref{eq:fid_bound}).
	\end{proof}

	An immediate corollary of~\cref{lem:upperbound_fidelity} is 
	\begin{equation}
	{\cal F}_{\rm ent}(U\!, {\cal S})\!\leq\! \frac{1}{N\!+\!1}\!\!\left(\!1\!+\!N
	\sqrt{1-\frac{\sum_j |\lambda_j|^2}{N}}\!+\! N \xi_{\rm max} \!\right),
	\end{equation}
	since $u(\mathcal{S})\leq 1$ for TPCP maps. However, this is in general not a very strong upper bound on the fidelity.

We can do better in the single-qubit case by realizing that there are strong relations between the singular values $\sigma_i$ of $T^{\mathcal{S}}$ and the absolute values of the eigenvalues $|\lambda_i|$ of $T^{\mathcal{S}}$. Ordering both the singular values and the eigenvalue magnitudes in descending order, we have the following (weak Majorization) inequalities for arbitrary matrices
	\begin{align}
	\prod_{i=1}^N \sigma_i &=  \prod_{i=1}^N |\lambda_i|,\\
	\sum_{i=1}^{F} \sigma_i &\geq  \sum_{i=1}^{F} |\lambda_i|, \;\;\;\; F\in \{1,\ldots,N-1\}.
	\end{align}
	For single-qubit channels we can also impose TPCP constraints to the singular values of the channel. In particular we have~\cite[Eq. (4)]{wolf2010inverse}
	\begin{gather}
	\sigma_i\leq 1,\;\; \forall i\in \{0,1,2,3\},\\ \sigma_1+\sigma_2\leq 1+ \sigma_3.
	\end{gather}
	Using these relations we can upper bound the unitarity of a single-qubit channel $\mathcal{S}$, given its eigenvalues, using the optimization:
	\begin{equation*}
	\begin{aligned}
	& \underset{\sigma_1,\sigma_2,\sigma_3}{\text{minimize}}
	& & u(\mathcal{S}) = \frac{1}{3}(\sigma_1^2+\sigma_2^2+\sigma_3^2) \\
	& \text{subject to}
	& & \sigma_1 \sigma_2 \sigma_3 = |\lambda_1||\lambda_2||\lambda_3|, \\
	&&& 1\geq \sigma_1 \geq \sigma_2 \geq \sigma_3 \geq 0,\\
	&&& \sigma_1+\sigma_2\leq 1+ \sigma_3,\\
	&&& \sigma_1+\sigma_2\geq |\lambda_1|+ |\lambda_2|,\\
	&&& \sigma_1+\sigma_2+\sigma_3\geq |\lambda_1|+ |\lambda_2| +|\lambda_3|.\\
	\end{aligned}
	\end{equation*}
	This is a non-convex optimization problem in three variables, for which a global minimum can be numerically computed given $\lambda_1,\lambda_2,\lambda_3$. This gives an upper bound on the unitarity of $\mathcal{S}$ and hence on the entanglement fidelity of $\mathcal{S}$ to the target unitary $U$. In the main text we use this optimization to give non-trivial upper bounds on the fidelities of single-qubit gates realized on superconducting chips and analyzed using the spectral tomography method.

\subsection{Frame Mismatch Accumulation}
	\label{sec:frame}
	In~\cref{sec:num} we noted that the data gathered for the CNOT gate cannot be explained by a model of a noisy TPCP map $\mathcal{S}$ repeated $k$ times. Here we propose a simple coherent model that qualitatively reproduces the features observed in~\cref{fig:CZ-quantum-exp} and we call this the frame mismatch accumulation model. Let $\mathcal{S}_0$ be a TPCP map that is a good approximation of the targeted gate applied exactly once (in the main text this was the CNOT) and let $V$ be some unitary. In the frame mismatch accumulation model we assume that $k$ consecutive applications of the gate are equal to:
	\begin{equation}
	\mathcal{S}_k = \prod_{i=0}^{k}\big(V^{\dagger}\big)^{i} \mathcal{S}_0 V^{i}= (V^{\dagger})^{k+1} (V{\cal S}_0)^k.
	\end{equation}
	Intuitively, this can be interpreted as an increasing mismatch between the frame in which $\mathcal{S}_0$ was defined and the frame in which the gate was implemented at the $i$-th repetition, up to $i=k$.
	
	We apply this model to a CNOT gate, choosing $\mathcal{S}_0$ to be an ideal CNOT gate and choosing $V = \exp(-i \frac{\theta}{2} I~\otimes~Y)$ with $\theta = 0.05$ deg.
	In the case of the cross-resonance CNOT gate performed on \emph{ibmqx4}, this may correspond to an imperfect cancellation of the $I\otimes Y$ term~\cite{Sheldon16}.
	In~\cref{fig:frame cnot} we see that this example closely reproduces the eigenvalues shown in~\cref{fig:CZ-quantum-exp}.
	At the same time, we note that the qualitative features observed in~\cref{fig:frame cnot} do not depend on the choice of the rotation axis of $V$ (for either qubit), as long as the rotation does not commute with $\mathcal{S}_0$ (which would leave the gate unaffected by the frame mismatch).
\begin{figure}
	\includegraphics[width=1.0\columnwidth]{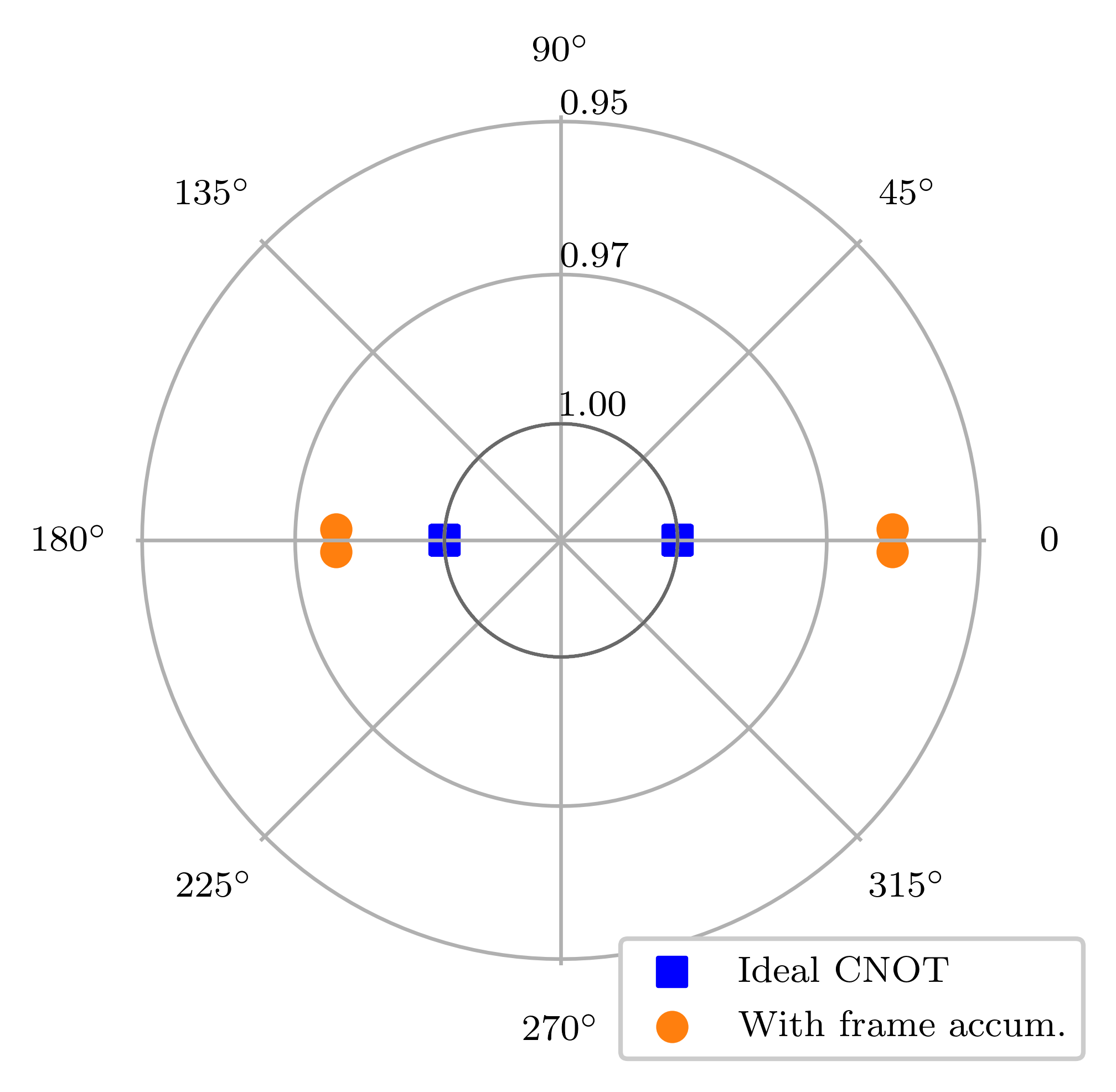}
	\caption{Spectral footprint of a simulated CNOT gate affected by frame mismatch accumulation, for $K=50$. The shown eigenvalues are 
	$\{0.9636+0.03276i,  0.9636-0.0327i, -0.9804+0.0495i , -0.9804-0.0495i\}$, qualitatively matching the experimentally-measured eigenvalues shown in~\cref{fig:CZ-quantum-exp} and, critically, matching the lack of real eigenvalues observed in~\cref{fig:CZ-quantum-exp}.  }
\label{fig:frame cnot}
	\end{figure}	
\section{Data Availability}
	Experimental data gathered for \cref{fig:pi4,fig:CZ-quantum-exp}, as well as an implementation of the matrix pencil algorithm can be found online at \url{https://doi.org/10.5281/zenodo.2613856}.
\section{Competing interests}
		The authors declare no competing interests. The views expressed in this manuscript are those of the authors and do not reflect the official policy or position of IBM or the IBM Quantum Experience Team
\section{Author contributions}
All authors contributed to the development of the theoretical concepts presented. The experiments on the IBM QE and the QuTech QI were performed and analysed by JH and the simulations on non-Markovianity were performed by JH and FB, under supervision of BMT. All authors contributed to the writing of the manuscript.
\section{Acknowledgements}
	The work by FB and BMT was supported by ERC grant EQEC No. 682726. JH is funded by STW Netherlands, NWO VIDI, an ERC Starting Grant and by the NWO Zwaartekracht QSC grant. We thank Andrew Cross for generous access to the IBM Quantum Experience for a TU Delft MSc project which indirectly led to this work, we thank Jarn de Jong for discussions on quantum tomography and Adriaan Rol for assistance with the QuTech Quantum Infinity. JH would also like to thank Arnaud Carignan-Dugas for interesting discussions.
\section{Additional information}
		The authors declare no competing interests. The views expressed in this manuscript are those of the authors and do not reflect the official policy or position of IBM or the IBM Quantum Experience Team.

\clearpage
\bibliographystyle{hunsrt}
\bibliography{tomo-refs.bib}

\end{document}